\documentclass[10pt]{IEEEtran}%
\usepackage{bbm}
\usepackage{epsfig}
\usepackage{easybmat}
\usepackage{graphics,graphicx,amssymb,amsmath,verbatim}
\usepackage{amssymb}
\usepackage{mathrsfs}
\usepackage{amsfonts}
\usepackage{amsmath}
\usepackage{graphicx}
\usepackage{multirow}
\usepackage{easybmat}
\usepackage{algorithm}
\usepackage{algorithmic}
\usepackage{inputenc}
\UseRawInputEncoding
\providecommand{\U}[1]{\protect \rule{.1in}{.1in}}
\newtheorem{theorem}{\quad Theorem}

\newtheorem{corollary}{\quad Corollary}

\newtheorem{remark}{\quad Remark}

\newenvironment{proof}[1][Proof]{\noindent \textbf{#1.} }{\  \rule{0.5em}{0.5em}}
\ifCLASSINFOpdf \else \fi
\allowdisplaybreaks[4]
\begin{document}

\title{A Robust Time-Delay Approach to Extremum Seeking via ISS Analysis of the Averaged System
\thanks{*This work was supported by the Planning and Budgeting Committee (PBC) Fellowship from the Council for Higher Education in Israel and by Israel Science Foundation (Grant No 673/19).}
}
\author{Xuefei Yang\thanks{Xuefei Yang and Emilia Fridman are with the School of Electrical Engineering,
 Tel-Aviv University, Israel (e-mail addresses: \texttt{xfyang1989@163.com, emilia@tauex.tau.ac.il}).}
\quad Emilia Fridman, \emph{Fellow, IEEE}
}
\date{}
\maketitle

\begin{abstract}
For N-dimensional (ND) static quadratic map, we present a time-delay approach to gradient-based extremum seeking (ES) both, in the continuous and, for the first time, the discrete domains. As in the recently introduced (for 2D maps in the continuous domain), we transform the system to the time-delay one (neutral type system in the form of Hale in the continuous case). This system is O($\varepsilon$)-perturbation of the averaged linear ODE system, where $\varepsilon$ is a period of averaging. We further explicitly present the neutral system as the linear ODE, where O($\varepsilon$)-terms are considered as disturbances with distributed delays of the length of the small parameter $\varepsilon$. Regional input-to-state stability (ISS) analysis is provided by employing a variation of constants formula that greatly simplifies the previously used analysis via Lyapunov-Krasovskii (L-K) method, simplifies the conditions and improves the results. Examples from the literature illustrate the efficiency of the new approach, allowing essentially large uncertainty of the Hessian matrix with bounds on $\varepsilon$ that are not too small.
\end{abstract}

\begin{IEEEkeywords}
Extremum seeking, Averaging, Time-delay, ISS.
\end{IEEEkeywords}

\IEEEpeerreviewmaketitle

\section{Introduction}

ES is a model-free, real-time on-line adaptive optimization control method. Under the premise of the existence of extremum value, the ES control can search the extremum value without relying on the prior knowledge of the input-output mapping relationship.
Because of its advantages of simple principle, low computational complexity and model free, ES control is used in many fields including anti-lock braking system \cite{xc20tvt,zo07tac}, aircraft formation flight \cite{bakb03jgcd}, maximum power point tracking of new energy generation such as solar \cite{hqt22spec}, wind \cite{so21ener} and fuel cells \cite{zamrg18tie}.

In 2000, Krstic and Wang gave the first rigorous stability analysis for an ES system by using averaging and singular perturbations in \cite{kw00auto}. Later on, this result was extended to the ES control for discrete-time systems \cite{ckal02tac}. Krstic's pioneer work laid a theoretical foundation for the research development of ES. Subsequently, a great amount of theoretical studies on ES are emerging. In \cite{tnm06auto,tnm09auto}, Tan et al. studied the non-local characteristics of perturbation ES, and extended the classical perturbation ES control to semi-global and global ES control. In \cite{lk10tac,lk16tac}, by combining the stochastic averaging theory with the ES theory, Liu and Krstic established a theoretical framework for stochastic ES in finite-dimensional space by selecting random signals as dither signals. In \cite{mmb10tac}, Moase et al. proposed a Newton-based ES algorithm, which can remove the dependence of the convergence rate on the unknown Hessian matrix. The Newton-based ES algorithm was later extended to the multi-variable case in \cite{gkn12auto}, which yields arbitrarily assignable convergence rates for each of the elements of the input vector. In \cite{dsej13auto}, Durr et al. introduced a novel interpretation of ES by using Lie bracket approximation, and shown that the Lie bracket system directly reveals the optimizing behavior of the ES system. In \cite{gd17auto}, Guay and Dochain proposed a proportional-integral extremum-seeking controller design technique that minimizes the impact of a time-scale separation on the transient performance of the ES controller. Recently, Oliveira et al. in \cite{okt17tac} first proposed a solution to the problem of designing multi-variable ES algorithms for delayed systems via standard predictors and backstepping transformation. Different from the standard prediction (which led to distributed terms in the control) used in \cite{okt17tac}, Malisoff et al. in \cite{mk21auto} used a one-stage sequential predictor approach to solve multi-variable ES problems with arbitrarily long input or output delays. Some more relevant research can be founded in the literature \cite{dksj17auto,hj17tac,mmwzs22arx,sut19auto}.

The conventional approach to analyze the stability of ES systems is dependent upon the classical averaging theory in finite dimensions (see \cite{khalil02book}) or infinite dimensions (see \cite{hl90jiea}). The basic idea is to approximate the original system by a simpler (averaged) system, namely, the practical stability of the original system can be guaranteed by the (asymptotically) stability of the averaged system, for sufficiently small parameter. However, these methods only provide the qualitative analysis, and cannot suggest quantitative upper bounds on the parameter that preserves the stability. Recently a new constructive time-delay approach to the continuous-time averaging was introduced in \cite{fz20auto}. This approach allows, for the first time, to derive efficient linear matrix inequality (LMI)-based conditions for finding the upper bound of the small parameter that ensures the stability. Later on, the time-delay approach to averaging was successfully applied to the quantitative stability  analysis of continuous-time ES algorithms (see \cite{zf22auto}) and sampled-data ES algorithms (see \cite{zfo22tac}) in the case of static maps by constructing appropriate Lyapunov--Krasovskii (L-K) functionals. However, the analysis via L-K method is complicated and the results are conservative, since only small uncertainties in Hessian and initial conditions are available.

In this paper, we suggest a robust time-delay approach to ES via ISS analysis of the averaged system both in the continuous and the discrete domains. After transforming the ES dynamics into a time-delay neutral type model as in \cite{zf22auto,zfo22tac}, we further transform it into an averaged ODE perturbed model, and then use the variation of constants formula instead of L-K method to quantitatively analyze the practical stability of the ODE system (and thus of the original ES system). Explicit conditions in terms of simple inequalities are established to guarantee the practical stability of the ES control systems. Through the solution of the constructed inequalities, we find upper bounds on the dither period that ensures the practical stability. Compared with the existing results, the main contribution of this paper and the significance of the obtained results can be stated as follows. First, comparatively to the considered continuous-time ES systems of one and two input variables in \cite{zf22auto,zfo22tac}, in the present paper we consider the N-variable case with arbitrary positive integer N, which is more general. Second, we develop, for the first time, the time-delay approach to ES control for discrete-time systems, and provide a quantitative analysis on the control parameters and the ultimate bound of seeking error. Third, comparatively to the L-K method utilized for neutral type systems in \cite{zf22auto,zfo22tac}, here we adopt the variation of constants formula for the ODE systems. This greatly simplifies the stability analysis process along with the stability conditions, and improve the quantitative bounds as well as the permissible range of the extremum value and the Hessian matrix. Moreover, our approach allows a larger decay rate and a smaller ultimate bound on the estimation error.

The paper¡¯s rest organization is as follows: In Section \ref{sec2} and Section \ref{sec3}, we apply the time-delay approach to the continuous-time ES and discrete-time ES, respectively. Each section contains two subsections: the theoretical results and examples with simulation results. Section \ref{sec4} concludes this paper.

\textbf{Notation:} The notation used in this article is fairly standard. For two integers $p$ and $q$ with $p\leq q,$ the notation $\mathbf{I}\left[p,q\right]  $ refers to the set $\left \{  p,p+1,\ldots,q\right \}  .$ The notations $\mathbf{N}_{+}$, $\mathbf{N}$ and $\mathbf{Z}$ refer to the set of positive integers, nonnegative integers and integers, respectively. The notation $P>0$ for $P\in \mathbf{R}^{n\times n}$ means that $P$ is symmetric and positive definite. The symmetric elements of the symmetric matrix are denoted by $\ast.$ The notations $\left \vert \cdot \right \vert $ and $\left \Vert \cdot \right \Vert $ refer to the usual Euclidean vector norm and the induced matrix $2$ norm, respectively.

\section{Continuous-Time ES} \label{sec2}

\subsection{A Time-Delay Approach to ES}

Consider the multi-variable static map given by%
\begin{equation}
\left.  y(t)=Q(\theta(t))=Q^{\ast}+\frac{1}{2}[\theta(t)-\theta^{\ast
}]^{\mathrm{T}}H[\theta(t)-\theta^{\ast}],\right.  \label{eq53}%
\end{equation}
where $y(t)\in \mathbf{R}$ is the measurable output, $\theta(t)\in
\mathbf{R}^{n}$ is the vector input, $Q^{\ast}\in \mathbf{R}$ and $\theta
^{\ast}\in \mathbf{R}^{n}$ are constants, $H=H^{\mathrm{T}}\in \mathbf{R}%
^{n\times n}$ is the Hessian matrix which is either positive definite or
negative definite. Without loss of generality, we assume that the static map
(\ref{eq53}) has a minimum value $y(t)=Q^{\ast}$ at $\theta(t)=\theta^{\ast},$
namely,%
\[
\left.  \left.  \frac{\partial Q}{\partial \theta}\right \vert _{\theta
=\theta^{\ast}}=0,\text{ }\left.  \frac{\partial^{2}Q}{\partial \theta^{2}%
}\right \vert _{\theta=\theta^{\ast}}=H>0.\right.
\]
Usually, the cost function is not known in (\ref{eq53}), but we can manipulate
$\theta(t)$. In the present paper, we assume that

\textbf{A1} The extremum point $\theta^{\ast}$ to be sought is uncertain from
a known ball where each of its elements satisfies $\theta_{i}^{\ast}%
(0)\in \lbrack \underline{\theta}_{i}^{\ast},\bar{\theta}_{i}^{\ast}]$
($i\in \mathbf{I[}1,n\mathbf{]}$) with $%
{\textstyle \sum \nolimits_{i=1}^{n}}
(\bar{\theta}_{i}^{\ast}-\underline{\theta}_{i}^{\ast})^{2}=\sigma_{0}^{2}.$

\textbf{A2} The extremum value $Q^{\ast}$ is unknown, but it is subject to
$\left \vert Q^{\ast}\right \vert \leq Q_{M}^{\ast}$ with $Q_{M}^{\ast}$ being known.

\textbf{A3} The Hessian matrix $H$ is unknown, but it is subject to $H=\bar
{H}+\Delta H$ with $\bar{H}>0$ being known and $\left \Vert \Delta H\right \Vert
\leq \kappa.$ Here $\kappa \geq0$ is a given scalar.

Under \textbf{A3}, there exist two positive scalars $H_{m}$ and $H_{M}$ such
that%
\begin{equation}
\left.  H_{m}\leq \left \Vert H\right \Vert \leq H_{M}.\right.  \label{eq71}%
\end{equation}

\begin{figure}[ptb]
\centering
\includegraphics[width=0.4\textwidth]{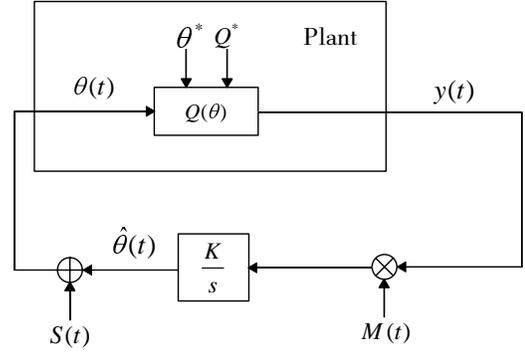}\caption{Extremum seeking
control scheme for continuous-time systems}%
\label{ES1}%
\end{figure}

The gradient-based classical ES algorithm depicted in Fig. \ref{ES1} is
governed by the following equations:%
\begin{equation}
\left.  \theta(t)=\hat{\theta}(t)+S(t),\text{ }\dot{\hat{\theta}%
}(t)=KM(t)y(t),\right.  \label{eq54}%
\end{equation}
where $\hat{\theta}(t)$ is the real-time estimate of $\theta^{\ast},$ $S(t)$
and $M(t)$ are the dither signals satisfying%
\begin{equation}
\left.
\begin{array}
[c]{l}%
S(t)=[a_{1}\sin(\omega_{1}t),\ldots,a_{n}\sin(\omega_{n}t)]^{\mathrm{T}},\\
M(t)=\left[  \frac{2}{a_{1}}\sin(\omega_{1}t),\ldots,\frac{2}{a_{n}}%
\sin(\omega_{n}t)\right]  ^{\mathrm{T}},
\end{array}
\right.  \label{eq54a}%
\end{equation}
in which $\omega_{i}\neq \omega_{j},i\neq j$ are non-zero, $\omega_{i}%
/\omega_{j}$ is rational and $a_{i}$ are real number. The adaptation gain $K$
is chosen as%
\[
K=\mathrm{diag}\{k_{1},k_{2},\ldots,k_{n}\},\text{ }k_{i}<0,\text{ }%
i\in \mathbf{I[}1,n\mathbf{]}%
\]
such that $KH$ (and also $K\bar{H}$) is Hurwitz (for instance, $K=kI_{n}$ with
a scalar $k<0$).

Define the estimation error $\tilde{\theta}(t)$ as%
\[
\tilde{\theta}(t)=\hat{\theta}(t)-\theta^{\ast}.
\]
Then by (\ref{eq54}), the estimation error is governed by%
\begin{equation}
\left.
\begin{array}
[c]{l}%
\dot{\tilde{\theta}}(t)=KM(t)\big[Q^{\ast}+\frac{1}{2}S^{\mathrm{T}%
}(t)HS(t)+\frac{1}{2}\tilde{\theta}^{\mathrm{T}}(t)H\tilde{\theta}(t)\\
+S^{\mathrm{T}}(t)H\tilde{\theta}(t)\big].
\end{array}
\right.  \label{eq17}%
\end{equation}
For the stability analysis of the ES control system (\ref{eq17}), several
methods are proposed in the existing literature including the classical
averaging approach (see \cite{ak03book,gkn12auto,kw00auto}), Lie brackets
approximation (see \cite{dsej13auto,lem22auto,sk17book}) and the recent time-delay
approach to averaging (see \cite{zf22auto,zfo22tac}). The classical averaging
approach usually resorts to the averaged system via the averaging theorem
\cite{khalil02book}. To be specific, treating $\tilde{\theta}(t)$ as a
\textquotedblleft freeze\textquotedblright \ constant in the averaging analysis
and defining $\omega_{i}=\frac{2\pi l_{i}}{\varepsilon},l_{i}\in \mathbf{N}_{+}$
($i\in \mathbf{I[}1,n\mathbf{]}$) satisfying $l_{i}\neq l_{j},i\neq j$, the
averaged system of (\ref{eq17}) can be derived as \cite{gkn12auto}
\begin{equation}
\dot{\tilde{\theta}}_{\mathrm{av}}(t)=KH\theta_{\mathrm{av}}(t), \label{eq10}%
\end{equation}
which is exponentially stable since $KH$ is Hurwitz.

The classical averaging approach leads to a qualitative analysis, namely, this
method cannot suggest quantitative lower bounds on the dither frequency that
guarantee the practical stability as well as the quantitative calculation of
the ultimate bound of seeking error. Recently, when the dimension $n=1,2$ in
(\ref{eq17}), motivated by \cite{fz20auto}, a constructive time-delay approach
for the stability analysis of gradient-based and bounded ES algorithms was
introduced in \cite{zf22auto,zfo22tac}. In the latter papers, the ES dynamics
was first converted into a time-delay neutral type model, and then the
L-K method was used to find sufficient practical stability
conditions in the form of LMIs.

Inspired by \cite{fz20auto,zf22auto}, we first apply the time-delay approach
to averaging of (\ref{eq17}). Integrating (\ref{eq17}) in $t\geq \varepsilon$
from $t-\varepsilon$ to $t,$ we get%
\begin{equation}
\left.
\begin{array}
[c]{l}%
\frac{1}{\varepsilon}%
{\textstyle \int \nolimits_{t-\varepsilon}^{t}}
\dot{\tilde{\theta}}(\tau)\mathrm{d}\tau=\frac{1}{\varepsilon}%
{\textstyle \int \nolimits_{t-\varepsilon}^{t}}
KM(\tau)Q^{\ast}\mathrm{d}\tau \\
+\frac{1}{2\varepsilon}%
{\textstyle \int \nolimits_{t-\varepsilon}^{t}}
KM(\tau)S^{\mathrm{T}}(\tau)HS(\tau)\mathrm{d}\tau \\
+\frac{1}{2\varepsilon}%
{\textstyle \int \nolimits_{t-\varepsilon}^{t}}
KM(\tau)\tilde{\theta}^{\mathrm{T}}(\tau)H\tilde{\theta}(\tau)\mathrm{d}\tau \\
+\frac{1}{\varepsilon}%
{\textstyle \int \nolimits_{t-\varepsilon}^{t}}
KM(\tau)S^{\mathrm{T}}(\tau)H\tilde{\theta}(\tau)\mathrm{d}\tau,\text{ }%
t\geq \varepsilon.
\end{array}
\right.  \label{eq47}%
\end{equation}
In the remainder of this paper, we define $x\pm y\triangleq x+y-y.$ For the
first term on the right-hand side of (\ref{eq47}), we have%
\begin{equation}
\left.
\begin{array}
[c]{l}%
\frac{1}{\varepsilon}%
{\textstyle \int \nolimits_{t-\varepsilon}^{t}}
KM(\tau)Q^{\ast}\mathrm{d}\tau \\
=\frac{1}{\varepsilon}Q^{\ast}K\operatorname{col}\big \{ \frac{2}{a_{i}}%
{\textstyle \int \nolimits_{t-\varepsilon}^{t}}
\sin \big(\frac{2\pi l_{i}}{\varepsilon}\tau \big)\mathrm{d}\tau \big \}_{i=1}%
^{n}\\
=0,
\end{array}
\right.  \label{eq48}%
\end{equation}
where we have used
\begin{equation}
\left.  \int \nolimits_{t-\varepsilon}^{t}\sin \big(\frac{2\pi l_{i}%
}{\varepsilon}\tau \big)\mathrm{d}\tau=0,\text{ }i\in \mathbf{I[}1,n\mathbf{]}%
.\right.  \label{eq48a}%
\end{equation}
For the second term on the right-hand side of (\ref{eq47}), we have%
\begin{equation}
\left.
\begin{array}
[c]{l}%
\frac{1}{2\varepsilon}%
{\textstyle \int \nolimits_{t-\varepsilon}^{t}}
KM(\tau)S^{\mathrm{T}}(\tau)HS(\tau)\mathrm{d}\tau \\
=\frac{1}{2\varepsilon}K%
{\textstyle \int \nolimits_{t-\varepsilon}^{t}}
{\textstyle \sum \limits_{i=1}^{n}}
\text{ }%
{\textstyle \sum \limits_{j=1}^{n}}
a_{i}a_{j}h_{ij}\sin \big(\frac{2\pi l_{i}}{\varepsilon}\tau \big)\sin
\big(\frac{2\pi l_{j}}{\varepsilon}\tau \big)M(\tau)\mathrm{d}\tau \\
=\frac{1}{\varepsilon}K\operatorname{col}\Big \{%
{\textstyle \sum \limits_{i=1}^{n}}
\text{ }%
{\textstyle \sum \limits_{j=1}^{n}}
\frac{a_{i}a_{j}h_{ij}}{a_{k}}%
{\textstyle \int \nolimits_{t-\varepsilon}^{t}}
\sin \big(\frac{2\pi l_{i}}{\varepsilon}\tau \big)\\
\times \sin \big(\frac{2\pi l_{j}}{\varepsilon}\tau \big)\sin \big(\frac{2\pi
l_{k}}{\varepsilon}\tau \big)\mathrm{d}\tau \Big \}_{k=1}^{n}\\
=0,
\end{array}
\right.  \label{eq49}%
\end{equation}
where we have utilized%
\[
\left.  \int \nolimits_{t-\varepsilon}^{t}\sin \Big(\frac{2\pi l_{i}%
}{\varepsilon}\tau \Big)\sin \Big(\frac{2\pi l_{j}}{\varepsilon}\tau
\Big)\sin \Big(\frac{2\pi l_{k}}{\varepsilon}\tau \Big)\mathrm{d}\tau=0.\right.
\]
For the third term on the right-hand side of (\ref{eq47}), we have%
\begin{equation}
\left.
\begin{array}
[c]{l}%
\frac{1}{2\varepsilon}%
{\textstyle \int \nolimits_{t-\varepsilon}^{t}}
KM(\tau)\tilde{\theta}^{\mathrm{T}}(\tau)H\tilde{\theta}(\tau)\mathrm{d}\tau \\
=\frac{1}{2\varepsilon}%
{\textstyle \int \nolimits_{t-\varepsilon}^{t}}
KM(\tau)[\tilde{\theta}^{\mathrm{T}}(\tau)H\tilde{\theta}(\tau)\pm
\tilde{\theta}^{\mathrm{T}}(t)H\tilde{\theta}(t)]\mathrm{d}\tau \\
=\frac{1}{2\varepsilon}\tilde{\theta}^{\mathrm{T}}(t)H\tilde{\theta}(t)K%
{\textstyle \int \nolimits_{t-\varepsilon}^{t}}
M(\tau)\mathrm{d}\tau \\
-\frac{1}{\varepsilon}%
{\textstyle \int \nolimits_{t-\varepsilon}^{t}}
{\textstyle \int \nolimits_{\tau}^{t}}
KM(\tau)\tilde{\theta}^{\mathrm{T}}(s)H\dot{\tilde{\theta}}(s)\mathrm{d}%
s\mathrm{d}\tau \\
=-\frac{1}{\varepsilon}%
{\textstyle \int \nolimits_{t-\varepsilon}^{t}}
{\textstyle \int \nolimits_{\tau}^{t}}
KM(\tau)\tilde{\theta}^{\mathrm{T}}(s)H\dot{\tilde{\theta}}(s)\mathrm{d}%
s\mathrm{d}\tau,
\end{array}
\right.  \label{eq50}%
\end{equation}
where we have employed $\int \nolimits_{t-\varepsilon}^{t}M(\tau)\mathrm{d}%
\tau=0$ via (\ref{eq48a}) and%
\[
\left.  \tilde{\theta}^{\mathrm{T}}(t)H\tilde{\theta}(t)-\tilde{\theta
}^{\mathrm{T}}(\tau)H\tilde{\theta}(\tau)=2%
{\textstyle \int \nolimits_{\tau}^{t}}
\tilde{\theta}^{\mathrm{T}}(s)H\dot{\tilde{\theta}}(s)\mathrm{d}s.\right.
\]
For the fourth term on the right-hand side of (\ref{eq47}), we have%
\begin{equation}
\left.
\begin{array}
[c]{l}%
\frac{1}{\varepsilon}%
{\textstyle \int \nolimits_{t-\varepsilon}^{t}}
KM(\tau)S^{\mathrm{T}}(\tau)H\tilde{\theta}(\tau)\mathrm{d}\tau \\
=\frac{1}{\varepsilon}%
{\textstyle \int \nolimits_{t-\varepsilon}^{t}}
KM(\tau)S^{\mathrm{T}}(\tau)H[\tilde{\theta}(\tau)\pm \tilde{\theta
}(t)]\mathrm{d}\tau \\
=\frac{1}{\varepsilon}K%
{\textstyle \int \nolimits_{t-\varepsilon}^{t}}
M(\tau)S^{\mathrm{T}}(\tau)\mathrm{d}\tau H\tilde{\theta}(t)\\
-\frac{1}{\varepsilon}%
{\textstyle \int \nolimits_{t-\varepsilon}^{t}}
{\textstyle \int \nolimits_{\tau}^{t}}
KM(\tau)S^{\mathrm{T}}(\tau)H\dot{\tilde{\theta}}(s)\mathrm{d}s\mathrm{d}%
\tau \\
=KH\tilde{\theta}(t)-\frac{1}{\varepsilon}%
{\textstyle \int \nolimits_{t-\varepsilon}^{t}}
{\textstyle \int \nolimits_{\tau}^{t}}
KM(\tau)S^{\mathrm{T}}(\tau)H\dot{\tilde{\theta}}(s)\mathrm{d}s\mathrm{d}\tau,
\end{array}
\right.  \label{eq51}%
\end{equation}
where we have utilized%
\[
\left.
{\textstyle \int \nolimits_{t-\varepsilon}^{t}}
M(\tau)S^{\mathrm{T}}(\tau)\mathrm{d}\tau=\varepsilon I_{n},\right.
\]
since%
\[
\left.  \int \nolimits_{t-\varepsilon}^{t}\frac{2a_{i}}{a_{j}}\sin
\Big(\frac{2\pi l_{i}}{\varepsilon}\tau \Big)\sin \Big(\frac{2\pi l_{j}%
}{\varepsilon}\tau \Big)\mathrm{d}\tau=\left \{
\begin{array}
[c]{cc}%
\varepsilon, & i=j,\\
0, & i\neq j.
\end{array}
\right.  \right.
\]
For the left-hand side of (\ref{eq47}), we have%
\begin{equation}
\left.  \frac{1}{\varepsilon}\int \nolimits_{t-\varepsilon}^{t}\dot
{\tilde{\theta}}(\tau)\mathrm{d}\tau=\frac{\mathrm{d}}{\mathrm{d}t}%
[\tilde{\theta}(t)-G(t)],\right.  \label{eq35}%
\end{equation}
where%
\begin{equation}
\left.  G(t)=\frac{1}{\varepsilon}\int \nolimits_{t-\varepsilon}^{t}%
(\tau-t+\varepsilon)\dot{\tilde{\theta}}(\tau)\mathrm{d}\tau.\right.
\label{eq37}%
\end{equation}
Finally, employing (\ref{eq48}), (\ref{eq49})-(\ref{eq35}), system (\ref{eq47}) can be
transformed to%
\begin{equation}
\left.  \frac{\mathrm{d}}{\mathrm{d}t}[\tilde{\theta}(t)-G(t)]=KH\tilde
{\theta}(t)-Y_{1}(t)-Y_{2}(t),\text{ }t\geq \varepsilon,\right.  \label{eq52b}%
\end{equation}
where%
\begin{equation}
\left.
\begin{array}
[c]{l}%
Y_{1}(t)=\frac{1}{\varepsilon}%
{\textstyle \int \nolimits_{t-\varepsilon}^{t}}
{\textstyle \int \nolimits_{\tau}^{t}}
KM(\tau)\tilde{\theta}^{\mathrm{T}}(s)H\dot{\tilde{\theta}}(s)\mathrm{d}%
s\mathrm{d}\tau,\\
Y_{2}(t)=\frac{1}{\varepsilon}%
{\textstyle \int \nolimits_{t-\varepsilon}^{t}}
{\textstyle \int \nolimits_{\tau}^{t}}
KM(\tau)S^{\mathrm{T}}(\tau)H\dot{\tilde{\theta}}(s)\mathrm{d}s\mathrm{d}\tau,
\end{array}
\right.  \label{eq56}%
\end{equation}
whereas $\dot{\tilde{\theta}}(s)$ is defined by the right-hand side of
(\ref{eq17}). Clearly, the solution $\tilde{\theta}(t)$ of system (\ref{eq17})
is also a solution of system (\ref{eq52b}). Thus, the practical stability of the original
non-delayed system (\ref{eq17}) can be guaranteed by the practical stability
of the time-delay system (\ref{eq52b}), which is a neutral type system with
the state $\tilde{\theta}$, as derived in \cite{zf22auto} for 2D maps.

In this paper, for simplifying the stability analysis, we further set
\begin{equation}
\left.  z(t)=\tilde{\theta}(t)-G(t).\right.  \label{eq52a}%
\end{equation}
Then system (\ref{eq52b}) can be rewritten as%
\begin{equation}
\left.  \dot{z}(t)=KHz(t)+KHG(t)-Y_{1}(t)-Y_{2}(t),\text{ }t\geq
\varepsilon.\right.  \label{eq52}%
\end{equation}

Comparatively to the averaged system (\ref{eq10}), system (\ref{eq52}) has the
additional terms $G(t),$ $Y_{1}(t)$ and $Y_{2}(t)$ that are of the order of
\textrm{O}$(\varepsilon)$ provided $\tilde{\theta}(s)$ and $\dot{\tilde
{\theta}}(s)$ (and thus $z(t)$) are of the order of
\textrm{O}$(1)$. Hence, for small $\varepsilon>0$
system (\ref{eq52}) can be regarded as a perturbation of system (\ref{eq10}).

Differently from \cite{zf22auto}, we will analyze (\ref{eq52}) as ODE w.r.t.
$z$ (and not as neutral type w.r.t. $\tilde{\theta}$) with delayed
disturbance-like \textrm{O}$(\varepsilon)$-terms $G,Y_{1},Y_{2}$ that depend
on the solutions of (\ref{eq17}). The resulting bound on $|z|$ will lead to
the bound on $\tilde{\theta}:|\tilde{\theta}|\leq|z|+|G|.$ The bound on $z$
will be found by utilizing the variation of constants formula compared to L-K
method employed in \cite{zf22auto}. This will greatly simplify the stability
analysis process along with the stability conditions, and improve the
quantitative bounds as well as the permissible range of the extremum value
$Q^{\ast}$ and the Hessian matrix $H$ in the numerical examples.

\begin{theorem}
\label{theorem1}Let \textbf{A1-A3} be satisfied. Consider the closed-loop
system (\ref{eq17}) with the initial condition $|\tilde{\theta}(0)|\leq
\sigma_{0}.$ Given tuning parameters $k_{i},$ $a_{i}$ ($i=\mathbf{I[}%
1,n\mathbf{]}$) and $\delta,$ let matrix $P$ ($I_{n}\leq P\leq pI_{n})$ with a
scalar $p\geq1$ and scalar $\zeta>0$ satisfy the following LMI:%
\begin{equation}
\left.
\begin{array}
[c]{l}%
\Phi_{1}=\left[
\begin{array}
[c]{cc}%
\Phi_{11} & PK\\
\ast & -\zeta I_{n}%
\end{array}
\right]  <0,\\
\Phi_{11}=\bar{H}^{\mathrm{T}}K^{\mathrm{T}}P+PK\bar{H}+2\delta P+\zeta
\kappa^{2}I_{n}.
\end{array}
\right.  \label{eq24a}%
\end{equation}
Given $\sigma>\sigma_{0}>0,$ let there exits $\varepsilon^{\ast}>0$ that
satisfy
\begin{equation}
\left.
\begin{array}
[c]{l}%
\Phi_{2}=p\left(  \sigma_{0}+\frac{\varepsilon^{\ast}\Delta \lbrack2(\Delta
_{1}+\Delta_{2}+\Delta_{3})+3\delta]}{2\delta}\right)  ^{2}\\
-\left(  \sigma-\frac{\varepsilon^{\ast}\Delta}{2}\right)  ^{2}<0,
\end{array}
\right.  \label{eq24}%
\end{equation}
where%
\begin{equation}
\left.
\begin{array}
[c]{l}%
\Delta=\left[  Q_{M}^{\ast}+\frac{H_{M}}{2}\left(  \sigma+\sqrt{%
{\textstyle \sum \nolimits_{i=1}^{n}}
a_{i}^{2}}\right)  ^{2}\right]  \sqrt{%
{\textstyle \sum \nolimits_{i=1}^{n}}
\frac{4k_{i}^{2}}{a_{i}^{2}}},\\
\Delta_{1}=\frac{H_{M}\max_{i\in \mathbf{I[}1,n\mathbf{]}}\left \vert
k_{i}\right \vert }{2},\text{ }\Delta_{2}=\frac{\sigma H_{M}}{2}\sqrt{%
{\textstyle \sum \nolimits_{i=1}^{n}}
\frac{4k_{i}^{2}}{a_{i}^{2}}},\\
\Delta_{3}=\frac{H_{M}}{2}\sqrt{%
{\textstyle \sum \nolimits_{i=1}^{n}}
\frac{4k_{i}^{2}}{a_{i}^{2}}}\sqrt{%
{\textstyle \sum \nolimits_{i=1}^{n}}
a_{i}^{2}}.
\end{array}
\right.  \label{eq70}%
\end{equation}
Then for all $\varepsilon \in(0,\varepsilon^{\ast}],$ the solution of the
estimation error system (\ref{eq17}) satisfies%
\begin{equation}
\left.
\begin{array}
[c]{l}%
|\tilde{\theta}(t)|<|\tilde{\theta}(0)|+\varepsilon \Delta<\sigma,\text{ }%
t\in \lbrack0,\varepsilon],\\
|\tilde{\theta}(t)|<\sqrt{p}\mathrm{e}^{-\delta(t-\varepsilon)}\left(
|\tilde{\theta}(0)|+\frac{3\varepsilon \Delta}{2}\right)  \\
+\frac{\varepsilon \Delta \lbrack2(\Delta_{1}+\Delta_{2}+\Delta_{3})\sqrt
{p}+\delta]}{2\delta}<\sigma,\text{ }t\geq \varepsilon.
\end{array}
\right.  \label{eq60}%
\end{equation}
Moreover, for all $\varepsilon \in(0,\varepsilon^{\ast}]$ and all initial
conditions $|\tilde{\theta}(0)|\leq \sigma_{0},$ the ball%
\begin{equation}
\left.  \Theta=\left \{  \tilde{\theta}\in \mathbf{R:}|\tilde{\theta}%
(t)|<\frac{\varepsilon \Delta \lbrack2(\Delta_{1}+\Delta_{2}+\Delta_{3})\sqrt
{p}+\delta]}{2\delta}\right \}  \right.  \label{eq69}%
\end{equation}
is exponential attractive with a decay rate $\delta.$
\end{theorem}

\begin{proof}
See Appendix A1.
\end{proof}

In the following, we make some explanations on Theorem \ref{theorem1}.

\begin{remark}
\label{remark2}Given any $\sigma_{0}$ and $\sigma^{2}>p\sigma_{0}^{2},$
inequality ($\Phi_{2}<0$ in (\ref{eq24})) is always feasible for small enough
$\varepsilon^{\ast}.$ Therefore, the result is semi-global. For $\Phi_{1}<0$
in (\ref{eq24a}), since $K\bar{H}$ is Hurwitz, there exists a $n\times n$
matrix $P>0$ such that for small enough $\delta>0,$ the following inequality
holds: $\Psi=\bar{H}^{\mathrm{T}}K^{\mathrm{T}}P+PK\bar{H}+2\delta P<0.$ We
choose $\zeta=1/\kappa.$ Applying the Schur complement to $\Phi_{1}<0,$ we
have%
\[
\left.  \Psi+\kappa \left(  I_{n}+PKKP\right)  <0,\right.
\]
which always holds for small enough $\kappa>0$ since $\Psi<0.$ For $P>0,$
there exist positive scalars $p_{1}$ and $p_{2}$ such that
\begin{equation}
\left.  p_{1}I_{n}\leq P\leq p_{2}I_{n}.\right.  \label{eq68}%
\end{equation}
If $p_{1}\neq1,$ we can rewrite (\ref{eq68}) as%
\[
\left.  I_{n}\leq \frac{1}{p_{1}}P\leq \frac{p_{2}}{p_{1}}I_{n},\right.
\]
which is in the form of $I_{n}\leq P\leq pI_{n}$ by setting $P=P/p_{1}$ and
$p=p_{2}/p_{1}.$ Furthermore, $\Phi_{i}<0$ ($i=1,2$) hold with the modified
\{$P,p$\} as well as the bound in (\ref{eq69}). The similar argument for the
LMI feasibility is applicable in Theorem \ref{theorem2}.
\end{remark}

\begin{remark}
\label{remark4}We give a brief discussion about the effect of free parameters on the
performance of ES system. For simplicity, let $K=kI_{n}$ with $k<0$ being a
given scalar. Then from (\ref{eq70}) we know that $\Delta$ and $\Delta_{i}$
($i\in \mathbf{I[}1,3\mathbf{]}$) are of the order of \textrm{O}$(\left \vert
k\right \vert )$ as well as the decay rate $\delta$ since $\delta=\left \vert k\right \vert \lambda_{\min}(H).$ Thus%
\[
\left.  \vartheta_{1}\triangleq \frac{2\sqrt{p}\Delta(\Delta_{1}+\Delta
_{2}+\Delta_{3})+\left(  3\sqrt{p}+1\right)  \Delta \delta}%
{2\delta}\right.
\]
is of the order of \textrm{O}$(\left \vert k\right \vert ).$ Note from
(\ref{eq24}) (which is equivalent to (\ref{eq24c})) that%
\[
\varepsilon^{\ast}<\frac{1}{\vartheta_{1}}\left(  \sigma-\sqrt{p}\sigma
_{0}\right)  ,
\]
which implies that for given $\sigma>\sigma_{0}>0,$ $\varepsilon^{\ast}$ is of
the order of \textrm{O}$(1/\left \vert k\right \vert ).$ Therefore, the decay rate $\delta$ increases as $\left \vert
k\right \vert $ increases, while $\varepsilon^{\ast}$ decreases as $\left \vert
k\right \vert $ increases. So we can adjust the gain $K=kI_{n}$ to balance the decay rate $\delta$ and $\varepsilon^{\ast}.$
In addition, we let%
\[
\left.  \vartheta_{2}\triangleq \frac{\Delta \lbrack2(\Delta_{1}+\Delta
_{2}+\Delta_{3})\sqrt{p}+\delta]}{2\delta}.\right.
\]
Then the ball in (\ref{eq69}) can be rewritten as%
\begin{equation}
\left.  \Theta=\left \{  \tilde{\theta}\in \mathbf{R:}|\tilde{\theta
}(t)|<\varepsilon \vartheta_{2}\right \}  .\right.  \label{eq69a}%
\end{equation}
Note from (\ref{eq70}) that $\vartheta_{2}$ is an increasing function of
$\sigma,$ thus, for given $\sigma_{0},$ $\delta,$ $\varepsilon,$ $a_{i}$ and
$k_{i}$ ($i\in \mathbf{I[}1,n\mathbf{]}$), we can solve the inequality
(\ref{eq24}) to find the smallest $\sigma,$ and then substitute it into
(\ref{eq69a}) to get the bound. Moreover, if $\varepsilon \vartheta_{2}%
<\sigma_{0}-\beta$ with some $\beta \in (0, \sigma_{0})$, we can reset $\sigma
_{0}=\varepsilon \vartheta_{2}+\beta$ and repeat the above process to obtain
a smaller ultimate bound (UB). Obviously, the lower bound of UB in theory is $\varepsilon
\vartheta_{2}$ with $\sigma=0.$
\end{remark}

\begin{remark}
\label{remark1}Compared with the results in \cite{zf22auto}, Theorem
\ref{theorem1} presents much simpler proof and LMI-based conditions, which
allow us to get larger decay rate and period of the dither signal. Moreover,
it is observed from (\ref{eq69}) that the ultimate bound on the estimation
error is of the order of $\mathrm{O}(\varepsilon)$ provided that $a_{i},k_{i}$
($i\in \mathbf{I[}1,n\mathbf{]}$) are of the order of $\mathrm{O}(1)$ leading
to $\delta$ of the order of $\mathrm{O}(1).$ This is smaller than
$\mathrm{O}(\sqrt{\varepsilon})$ achieved in \cite{zf22auto}. In addition, due
to the complexity of the LMIs in the vertices when the Hessian $H$ is not
known, the work \cite{zf22auto} did not go into details to discuss the
uncertainty case. As a comparison, by using the established time-delay
approach, we can easily solve the uncertainty case.
\end{remark}

\begin{remark}
We have taken the backward averaged method (\textquotedblleft
backward\textquotedblright \ refers to the interval $[t-\varepsilon,t]$ rather
than $[t,t+\varepsilon]$) to derive the ODE system as shown in
(\ref{eq52}). If the forward averaged method is adopted, namely, integrating
(\ref{eq17}) in $t\geq0$ from $t$ to $t+\varepsilon,$ we will obtain the
following closed-loop vector system:%
\begin{equation}
\left.  \dot{z}(t)=KHz(t)+KHG(t)+Y_{1}(t)+Y_{2}(t),\text{ }t\geq0,\right.
\label{eq33}%
\end{equation}
where%
\[
\left.
\begin{array}
[c]{l}%
z(t)=\tilde{\theta}(t)-G(t),\\
G(t)=\frac{1}{\varepsilon}%
{\textstyle \int \nolimits_{t}^{t+\varepsilon}}
(\tau-t-\varepsilon)\dot{\tilde{\theta}}(\tau)\mathrm{d}\tau,\\
Y_{1}(t)=\frac{1}{\varepsilon}%
{\textstyle \int \nolimits_{t}^{t+\varepsilon}}
{\textstyle \int \nolimits_{t}^{\tau}}
KM(\tau)\tilde{\theta}^{\mathrm{T}}(s)H\dot{\tilde{\theta}}(s)\mathrm{d}%
s\mathrm{d}\tau,\\
Y_{2}(t)=\frac{1}{\varepsilon}%
{\textstyle \int \nolimits_{t}^{t+\varepsilon}}
{\textstyle \int \nolimits_{t}^{\tau}}
KM(\tau)S^{\mathrm{T}}(\tau)H\dot{\tilde{\theta}}(s)\mathrm{d}s\mathrm{d}\tau.
\end{array}
\right.
\]
System (\ref{eq33}) is of advanced type as it depends on the future values of
$\tilde{\theta}(s)$ and $\dot{\tilde{\theta}}(s),$ $s\in \lbrack
t,t+\varepsilon].$ Note that system (\ref{eq33}) is available from $t\geq0$
rather than $t\geq \varepsilon$ for (\ref{eq52}), it seems that by using
arguments of Theorem \ref{theorem1} for (\ref{eq33}), we will get a better
result. However, it is not well-posed with the initial conditions at $t=0$ and seems to be impossible to prove the assumption
$|\tilde{\theta}(t)|<\sigma,$ $\forall t\geq0$ because of the advanced
information. This is the reason that we take the backward integration instead
of the forward one here.
\end{remark}

Next we consider a special case with the Hessian $H$ being diagonal, namely,
$H=\mathrm{diag}\{h_{1},h_{2},\ldots,h_{n}\}$ with $h_{i}>0,i\in
\mathbf{I[}1,n\mathbf{].}$ We also assume that $H$ is unknown, but satisfies
(\ref{eq71}). In this case, instead of utilizing the Lyapunov method to find
the upper bound of the fundamental matrix $\mathrm{e}^{KHt},$ we can directly
compute that%
\[
\left.  \left \Vert \mathrm{e}^{KHt}\right \Vert \leq \mathrm{e}^{-H_{m}%
\min_{i\in \mathbf{I[}1,n\mathbf{]}}\left \vert k_{i}\right \vert t}%
\triangleq \mathrm{e}^{-\delta t},\text{ }\forall t\geq0.\right.
\]
This can lead to a simpler analysis and more concise result as shown in the
following corollary.

\begin{corollary}
\label{corollary1}Let \textbf{A1-A2} be satisfied and the diagonal Hessian $H$
be unknown but satisfy (\ref{eq71}). Consider the closed-loop system
(\ref{eq17}) with the initial condition $|\tilde{\theta}(0)|\leq \sigma_{0}.$
Given tuning parameters $k_{i},$ $a_{i}$ ($i=\mathbf{I[}1,n\mathbf{]}$) and
$\sigma>\sigma_{0}>0,$ let there exits $\varepsilon^{\ast}>0$ that satisfy
\[
\left.  \Phi=\sigma_{0}+\frac{\varepsilon^{\ast}\Delta \lbrack(\Delta
_{1}+\Delta_{2}+\Delta_{3})+2\delta]}{\delta}<\sigma,\right.
\]
where $\Delta$ and $\Delta_{i}$ $(i=\mathbf{I[}1,3\mathbf{]})$ are given by
(\ref{eq70}). Then for all $\varepsilon \in(0,\varepsilon^{\ast}],$ the
solution of the estimation error system (\ref{eq17}) satisfies%
\[
\left.
\begin{array}
[c]{l}%
|\tilde{\theta}(t)|<|\tilde{\theta}(0)|+\varepsilon \Delta<\sigma,\text{ }%
t\in \lbrack0,\varepsilon],\\
|\tilde{\theta}(t)|<\mathrm{e}^{-\delta(t-\varepsilon)}\left(  |\tilde{\theta
}(0)|+\frac{3\varepsilon \Delta}{2}\right)  \\
+\frac{\varepsilon \Delta \lbrack2(\Delta_{1}+\Delta_{2}+\Delta_{3})+\delta
]}{2\delta}<\sigma,\text{ }t\geq \varepsilon.
\end{array}
\right.
\]
Moreover, for all $\varepsilon \in(0,\varepsilon^{\ast}]$ and all initial
conditions $ \vert \tilde{\theta}(0) \vert \leq \sigma_{0},$ the ball%
\[
\left.  \Theta=\left \{  \tilde{\theta}\in \mathbf{R:}|\tilde{\theta}%
(t)|<\frac{\varepsilon \Delta \lbrack2(\Delta_{1}+\Delta_{2}+\Delta_{3}%
)+\delta]}{2\delta}\right \}  \right.
\]
is exponential attractive with a decay rate $\delta=H_{m}\min_{i\in
\mathbf{I[}1,n\mathbf{]}}\left \vert k_{i}\right \vert .$
\end{corollary}

\begin{remark}
\label{remark3}When $n=1,$ the results in Corollary \ref{corollary1} can be
further improved as follows. We note that%
\[
\left.
\begin{array}
[c]{l}%
|G(t)|<\frac{\varepsilon \Delta}{2},\text{ }|KHG(t)|<\frac{\varepsilon
\Delta \cdot \left \vert KH\right \vert }{2},\text{ }|Y_{1}(t)|<\frac
{\varepsilon \Delta \cdot \sigma \left \vert KH\right \vert }{\left \vert
a\right \vert },\\
|Y_{2}(t)|<\varepsilon \Delta \cdot \left \vert KH\right \vert ,\text{ }%
\mathrm{e}^{KHt}\leq \mathrm{e}^{-\left \vert KH_{m}\right \vert t}%
\triangleq \mathrm{e}^{-\delta t},
\end{array}
\right.
\]
where%
\[
\Delta=\left[  Q_{M}^{\ast}+\frac{H_{M}}{2}\left(  \sigma+\left \vert
a\right \vert \right)  ^{2}\right]  \frac{2\left \vert K\right \vert }{\left \vert
a\right \vert }.
\]
Then via (\ref{eq22}), we have%
\[
\left.
\begin{array}
[c]{l}%
|z(t)|<\mathrm{e}^{-\delta \left(  t-\varepsilon \right)  }|z(\varepsilon)|\\
+\frac{\varepsilon \Delta(3\left \vert a\right \vert +2\sigma)\left \vert
KH\right \vert }{2\left \vert a\right \vert }%
{\textstyle \int \nolimits_{\varepsilon}^{t}}
\mathrm{e}^{-\left \vert KH\right \vert (t-s)}\mathrm{d}s\\
=\mathrm{e}^{-\delta \left(  t-\varepsilon \right)  }|z(\varepsilon)|\\
+\frac{\varepsilon \Delta(3\left \vert a\right \vert +2\sigma)}{2\left \vert
a\right \vert }\left(  1-\mathrm{e}^{-\left \vert KH\right \vert (t-\varepsilon
)}\right)  \\
\leq \mathrm{e}^{-\delta \left(  t-\varepsilon \right)  }|z(\varepsilon
)|+\frac{\varepsilon \Delta(3\left \vert a\right \vert +2\sigma)}{2\left \vert
a\right \vert }.
\end{array}
\right.
\]
It follows that%
\[
\left.
\begin{array}
[c]{l}%
\left \vert \tilde{\theta}(t)\right \vert <\mathrm{e}^{-\delta(t-\varepsilon
)}\left(  |\tilde{\theta}(0)|+\frac{3\varepsilon \Delta}{2}\right)
+\frac{\varepsilon \Delta(3\left \vert a\right \vert +2\sigma)}{2\left \vert
a\right \vert }+\frac{\varepsilon \Delta}{2}\\
=\mathrm{e}^{-\delta(t-\varepsilon)}\left(  |\tilde{\theta}(0)|+\frac
{3\varepsilon \Delta}{2}\right)  +\frac{\varepsilon \Delta(2\left \vert
a\right \vert +\sigma)}{\left \vert a\right \vert },\text{ }t\geq \varepsilon,
\end{array}
\right.
\]
which implies%
\[
\left.  \left \vert \tilde{\theta}(t)\right \vert <\mathrm{e}^{-\delta
(t-\varepsilon)}\left(  |\tilde{\theta}(0)|+\frac{3\varepsilon \Delta}%
{2}\right)  +\frac{\varepsilon \Delta(2\left \vert a\right \vert +\sigma
)}{\left \vert a\right \vert }<\sigma,\right.
\]
if%
\[
\left.  \Phi=\sigma_{0}+\frac{\varepsilon^{\ast}\Delta(7\left \vert
a\right \vert +2\sigma)}{2\left \vert a\right \vert }<\sigma,\text{ }%
\forall \varepsilon \in(0,\varepsilon^{\ast}].\right.
\]
In this case, the ultimate bound is given by%
\[
\left.  \Theta=\left \{  \tilde{\theta}\in \mathbf{R:}|\tilde{\theta}%
(t)|<\frac{\varepsilon \Delta(2\left \vert a\right \vert +\sigma)}{\left \vert
a\right \vert }\right \}  \right.
\]
with a decay rate $\delta=\left \vert KH_{m}\right \vert .$
\end{remark}

\subsection{Examples}

\subsubsection{Scalar systems}

Consider the single-input map \cite{zf22auto}%
\[
Q(\theta(t))=Q^{\ast}+\frac{H}{2}\theta^{2}(t)
\]
with $Q^{\ast}=0$ and $H=2.$ Note that%
\[
\left.
\begin{array}
[c]{ll}%
\dot{\hat{\theta}}(t)=Ka\sin(\omega t)y(t) & \text{in \cite{zf22auto},}\\
\dot{\hat{\theta}}(t)=\frac{2K}{a}\sin(\omega t)y(t) & \text{in Remark
\ref{remark3},}%
\end{array}
\right.
\]
then for a fair comparison, we select the tuning parameters of the
gradient-based ES as%
\[
\left.
\begin{array}
[c]{ll}%
a=0.1,\text{ }K=-1.3 & \text{in \cite{zf22auto},}\\
a=0.1,\text{ }K=-1.3\cdot \frac{0.1^{2}}{2}=-6.5\mathrm{e}^{-3} & \text{in
Remark \ref{remark3}.}%
\end{array}
\right.
\]
If $Q^{\ast}$ and $H$ are unknown, but satisfy \textbf{A2} and (\ref{eq71}) we
consider%
\[
\left.
\begin{array}
[c]{c}%
Q_{M}^{\ast}=0.1,\text{ }H_{m}=1.9,\text{ }H_{M}=2.1;\\
Q_{M}^{\ast}=1.0,\text{ }H_{m}=1.5,\text{ }H_{M}=8.0.
\end{array}
\right.
\]
Both the solutions of uncertainty-free and uncertainty cases are shown in
Table \ref{tab1}. By comparing the data, we find that our results in Remark
\ref{remark3} allows larger decay rate $\delta$ and upper bound $\varepsilon
^{\ast}$ than those in \cite{zf22auto}.
Moreover, when the upper bound
$\varepsilon^{\ast}$ shares the same value, our results allow much larger
uncertainties in initial condition $\sigma_{0},$ extremum value $Q_{M}^{\ast}$
and Hessian matrix $H$ than those in \cite{zf22auto}. Finally, we make a
comparison for the UB by using Remark \ref{remark4}. For a fair comparison, we choose the same value of
$\varepsilon.$ Both the solutions of uncertainty-free and uncertainty cases
are shown in Table \ref{tab2}. It follows that our results allow much smaller
values of UB than those in \cite{zf22auto}.

\begin{table}[h]
\caption{Comparison of $\varepsilon^{\ast}$ in scalar systems}%
\label{tab1}%
\center
\setlength{\tabcolsep}{3pt}
\begin{tabular}
[c]{|l|c|c|c|c|}\hline
ES: sine wave & $\sigma_{0}$ & $\sigma$ & $\delta$ & $\varepsilon^{\ast}%
$\\ \hline
\cite{zf22auto} ($Q^{\ast}=0,H=2$) & 1 & $\sqrt{2}$ & 0.010 & 0.021\\ \hline
Remark \ref{remark3} ($Q^{\ast}=0,H=2$) & 1 & $\sqrt{2}$ & 0.013 &
0.079\\ \hline
Remark \ref{remark3} ($Q^{\ast}=0,H=2$) & 2.14 &
3.30 & 0.013 & 0.021\\ \hline
\cite{zf22auto} ($\big \vert Q^{\ast}\big \vert \leq0.1,$ $1.9\leq
\big \vert H\big \vert \leq2.1$) & 1 & $\sqrt{2}$ & 0.010 & 0.018\\ \hline
Remark \ref{remark3} ($\big \vert Q^{\ast}\big \vert \leq0.1,$ $1.9\leq
\big \vert H\big \vert \leq2.1$) & 1 & $\sqrt{2}$ & 0.012 & 0.072\\ \hline
Remark \ref{remark3} ($\big \vert Q^{\ast}\big \vert \leq1.0,$ $1.6\leq
\big \vert H\big \vert \leq7.9$) & 1 & $\sqrt{2}$ & 0.010 & 0.018\\ \hline
\end{tabular}
\end{table}

\begin{table}[h]
\caption{Comparison of UB in scalar systems}%
\label{tab2}%
\center  {\scriptsize { \setlength{\tabcolsep}{3pt}
\begin{tabular}
[c]{|l|c|c|c|c|c|}\hline
ES: sine wave& $\sigma_{0}$ & $\sigma$ & $\delta$ & $\varepsilon$ & UB\\ \hline
\cite{zf22auto} ($Q^{\ast}=0,H=2$) & 1 & $\sqrt{2}$ & 0.010 &
0.021 & 0.68\\ \hline
Remark 4 ($Q^{\ast}=0,H=2$) & 2.14 & 3.30 & 0.013 & 0.021 &
$1.9\mathrm{e}^{-4}$\\ \hline
\cite{zf22auto} ($\big \vert Q^{\ast}\big \vert \leq0.1,$ $1.9\leq
\big \vert H\big \vert \leq2.1$) & 1 & $\sqrt{2}$ & 0.010 & 0.018 &
0.71\\ \hline
Remark 4 ($\big \vert Q^{\ast}\big \vert \leq1.0,$ $1.6\leq
\big \vert H\big \vert \leq7.9$) & 1 & $\sqrt{2}$ & 0.010 & 0.018 & $5.3\mathrm{e}^{-3}$\\ \hline
\end{tabular}
} }\end{table}

For the numerical simulations, we choose $\omega=2\pi/\varepsilon$ and the
same other parameter values as shown in second and fourth rows in Table
\ref{tab2} for the uncertainty-free and uncertainty cases, respectively. In
addition, in the uncertainty case, we let
\[
\left.  H=4.75+3.15\sin t\right.  ,
\]
which satisfies the condition $1.6\leq \left \vert H\right \vert \leq7.9$ as
shown in Table \ref{tab2}. Under the initial condition $\hat{\theta}(0)=2$ for the
uncertainty-free case and $\hat{\theta}(0)=1$ for the uncertainty case,
the simulation results are shown in Fig. \ref{c1c} and Fig. \ref{c1u},
respectively, from which we can see that the values of UB shown in Table
\ref{tab2} are confirmed.

\begin{figure}[ptb]
\centering
\includegraphics[width=0.45\textwidth]{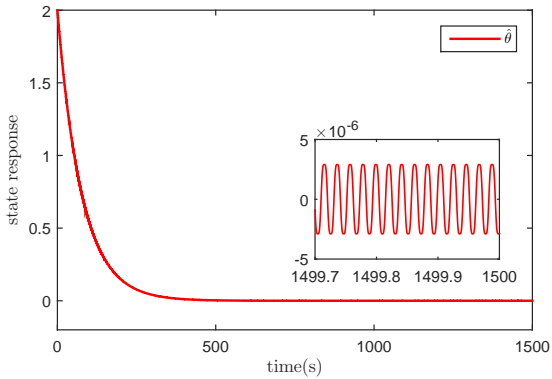}\caption{Trajectory of the
real-time estimate $\hat{\theta}$ for the uncertainty-free case}%
\label{c1c}%
\end{figure}

\begin{figure}[ptb]
\centering
\includegraphics[width=0.45\textwidth]{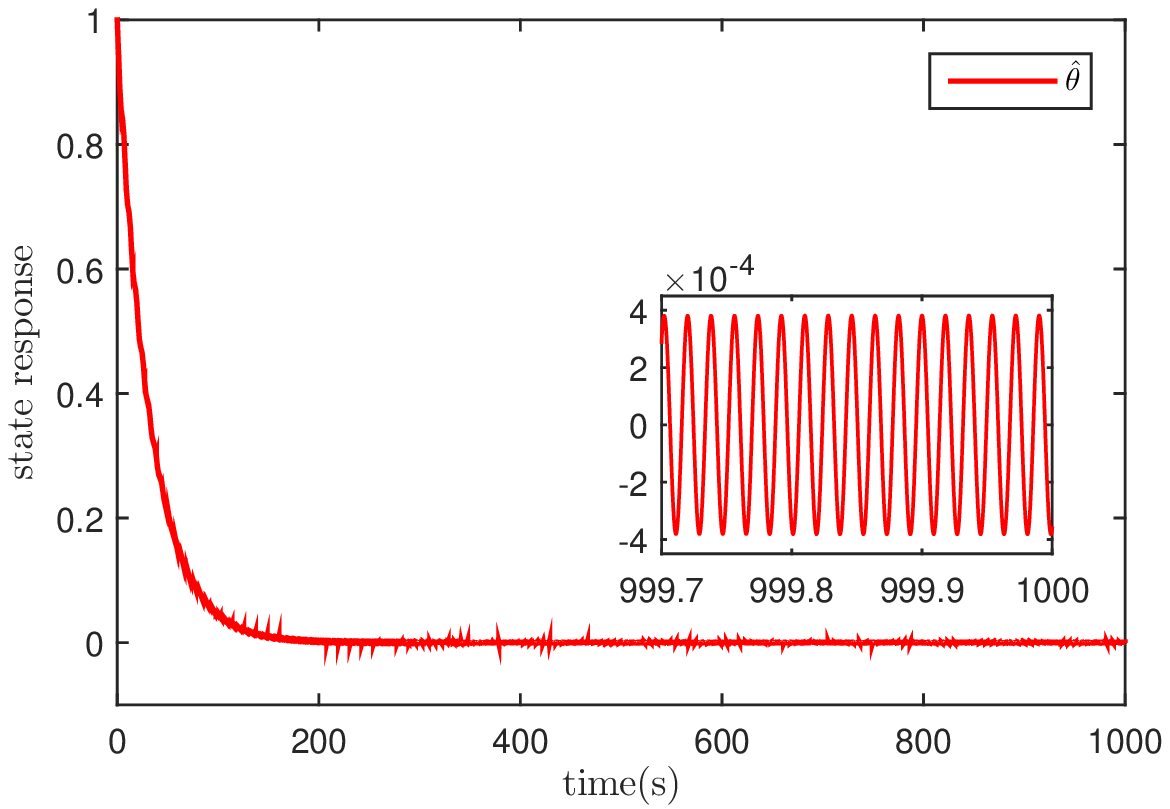}\caption{Trajectory of the
real-time estimate $\hat{\theta}$ for the uncertainty case}%
\label{c1u}%
\end{figure}

\subsubsection{Vector systems: $n=2$}

Consider an autonomous vehicle in an environment without GPS orientation
\cite{sk17book}. The goal is to reach the location of the stationary minimum
of a measurable function%
\[
\left.
\begin{array}
[c]{l}%
J(x(t),y(t))=Q^{\ast}+\frac{1}{2}\left[
\begin{array}
[c]{cc}%
x(t) & y(t)
\end{array}
\right]  H\left[
\begin{array}
[c]{c}%
x(t)\\
y(t)
\end{array}
\right] \\
=x^{2}(t)+y^{2}(t),
\end{array}
\right.
\]
where%
\[
\left.  Q^{\ast}=0,\text{ }H=\left[
\begin{array}
[c]{cc}%
2 & 0\\
0 & 2
\end{array}
\right]  .\right.
\]
We employ the classical ES%
\[
\left.
\begin{array}
[c]{l}%
x(t)=\hat{x}(t)+a_{1}\sin(\omega_{1}t),\text{ }y(t)=\hat{y}(t)+a_{2}%
\sin(\omega_{2}t),\\
\dot{\hat{x}}(t)=\frac{2k_{1}}{a_{1}}\sin(\omega_{1}t)J(t),\text{ }\dot{\hat{y}}%
(t)=\frac{2k_{2}}{a_{2}}\sin(\omega_{2}t)J(t)
\end{array}
\right.
\]
with $k_{1}=k_{2}=-0.01,$ $a_{1}=a_{2}=0.2$. The solutions are shown in Table
\ref{tab3}. It follows that Corollary \ref{corollary1} allows larger decay
rate $\delta$ and much larger upper bound $\varepsilon^{\ast}$ than those in
\cite{zf22auto}. Moreover, when the upper bound $\varepsilon^{\ast}$ shares the same value, our results allow much larger uncertainties in initial condition $\sigma_{0}$ than those in \cite{zf22auto}. Finally, we make a
comparison for the ultimate bound under the same value of $\varepsilon.$ The
solutions are shown in Table \ref{tab5}. It follows that the values of UB
obtained by Corollary \ref{corollary1} are much smaller that those in
\cite{zf22auto}.

\begin{table}[h]
\caption{Comparison of $\varepsilon^{\ast}$ in vector systems: $n=2$}%
\label{tab3}%
\center
\setlength{\tabcolsep}{3pt}
\begin{tabular}
[c]{|l|c|c|c|c|}\hline
ES: sine wave & $\sigma_{0}$ & $\sigma$ & $\delta$ & $\varepsilon^{\ast}%
$\\ \hline
\cite{zf22auto} & $\sqrt{2}$ & $2\sqrt{2}$ & 0.01 & 0.017\\ \hline
Corollary \ref{corollary1} & $\sqrt{2}$ & $2\sqrt{2}$ & 0.02 & 0.042\\ \hline
Corollary \ref{corollary1} & 2.55 & 4 &
0.02 & 0.017\\ \hline
\end{tabular}
\end{table}

\begin{table}[h]
\caption{Comparison of UB in vector systems: $n=2$}%
\label{tab5}%
\center  {\setlength{\tabcolsep}{3pt}
\begin{tabular}
[c]{|l|c|c|c|c|c|}\hline
ES: sine wave& $\sigma_{0}$ & $\sigma$ & $\delta$ & $\varepsilon$ & UB\\ \hline
\cite{zf22auto}  & $\sqrt{2}$ & 2$\sqrt{2}$ & 0.01 &
0.017 & 1.9\\ \hline
Corollary \ref{corollary1} & 2.55 & 4 & 0.02 & 0.017 & $1.4\mathrm{e}^{-3}$\\ \hline
\end{tabular}
}\end{table}

For the numerical simulations, we choose the same parameter values as in
Corollary \ref{corollary1} in Table \ref{tab5} and $\omega_{2}=2\omega_{1},$
$\omega_{1}=2\pi/\varepsilon.$ Under the initial condition $\hat{x}(0)=2,$
$\hat{y}(0)=2,$ the simulation results are shown in Fig. \ref{c2c}, from which
we can see that the value of UB shown in Table \ref{tab5} is confirmed.

\begin{figure}[ptb]
\centering
\includegraphics[width=0.45\textwidth]{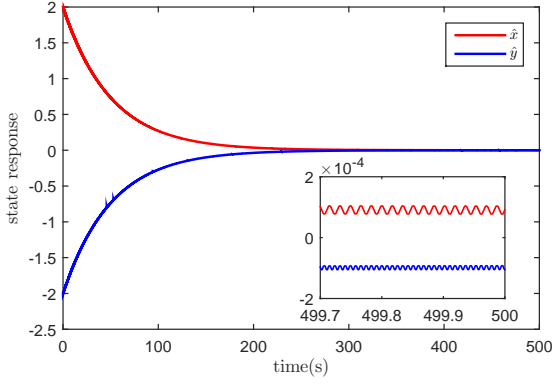}\caption{Trajectories of the
real-time estimate ($\hat{x},\hat{y}$)}%
\label{c2c}%
\end{figure}

\subsubsection{Vector systems: $n=6$}

Consider the quadratic function (\ref{eq53}) with \cite{dsej13auto}%
\[
Q^{\ast}=0,\text{ }\theta^{\ast}=[1,1,-1,-1,-1,1]^{\mathrm{T}},\text{
}H=\mathrm{diag}\{1,1,1,1,1,3\}.
\]
If $Q^{\ast}$ and $H$ are unknown, but satisfy \textbf{A2} and (\ref{eq71}) we
consider%
\[
\left.  Q_{M}^{\ast}=0.5,\text{ }H_{m}=0.8,\text{ }H_{M}=3.2.\right.
\]
We select the tuning parameters of the gradient-based ES as $k_{i}=-0.05,$
$a_{i}=1,$ $i\in \mathbf{I[}1,6\mathbf{].}$ The solutions are shown in Table
\ref{tab7}.

\begin{table}[h]
\caption{Vector systems: $n=6$}%
\label{tab7}
\center { \setlength{\tabcolsep}{3pt}
\begin{tabular}
[c]{|l|c|c|c|c|c|}\hline
ES: sine wave& $\sigma_{0}$ & $\sigma$ & $\delta$ & $\varepsilon^{\ast}$ & UB\\ \hline
Uncertainty-free case & 1 & $2$ & 0.150 & $1.0\mathrm{e}^{-2}$ & 0.315\\ \hline
Uncertainty case & 1 & $2$ & 0.025 & $1.4\mathrm{e}^{-3}$ &
0.382\\ \hline
\end{tabular}
}\end{table}

\section{Discrete-Time ES} \label{sec3}

In this section, we will establish the time-delay approach for discrete-time
ES. Although some arguments are similar to continuous-time ES, it is important
to present the discrete-time results by noting that results for discrete-time
ES are not as readily available as their continuous counterparts, and the
derivation is not straightforward.

\subsection{A Time-Delay Approach to ES}

Consider multi-variable static maps given by \cite{fkb13ejc}%
\begin{equation}
\left.  y(k)=Q(\theta(k))=Q^{\ast}+\frac{1}{2}[\theta(k)-\theta^{\ast
}]^{\mathrm{T}}H[\theta(k)-\theta^{\ast}],\right.  \label{eq119}%
\end{equation}
where $y(k)\in \mathbf{R}$ is the measurable output, $\theta(k)\in
\mathbf{R}^{n}$ is the vector input, $H=H^{\mathrm{T}}\in \mathbf{R}^{n\times
n}$ is the Hessian matrix which, without loss of generality, is positive
definite. In the present paper, we also assume that $\theta^{\ast},$ $Q^{\ast
}$ and $H$ satisfy \textbf{A1-A3}. The gradient-based classical ES algorithm
depicted in Fig. \ref{ES2} is designed as follows \cite{fkb13ejc}:%
\begin{equation}
\left.
\begin{array}
[c]{l}%
\theta \left(  k\right)  =\hat{\theta}\left(  k\right)  +S(k),\\
\hat{\theta}\left(  k+1\right)  =\hat{\theta}\left(  k\right)  +\varepsilon
LM(k)y(k),
\end{array}
\right.  \label{eq120}%
\end{equation}
where%
\begin{equation}
\left.
\begin{array}
[c]{l}%
S(k)=[a_{1}\sin(\omega_{1}k),\ldots,a_{n}\sin(\omega_{n}k)]^{\mathrm{T}},\\
M(k)=\left[  \frac{2}{a_{1}}\sin(\omega_{1}k),\ldots,\frac{2}{a_{n}}%
\sin(\omega_{n}k)\right]  ^{\mathrm{T}},
\end{array}
\right.  \label{eq120a}%
\end{equation}
in which $\omega_{i}=b_{i}\pi$ $(i\in \mathbf{I[}1,n\mathbf{]})$ with
$\left \vert b_{i}\right \vert \in(0,1)$ being a rational number and $\omega
_{i}\neq \omega_{j},i\neq j$ (see \cite{fkb13ejc}). The adaptation gain $L$\ is
chosen as%
\[
L=\mathrm{diag}\{l_{1},l_{2},\ldots,l_{n}\},\text{ }l_{i}<0,\text{ }%
i\in \mathbf{I[}1,n\mathbf{],}%
\]
such that $LH$ (and also $L\bar{H}$) being Hurwitz (for instance, $L=lI_{n}$
with a scalar $l<0$).

\begin{figure}[ptb]
\centering
\includegraphics[width=0.4\textwidth]{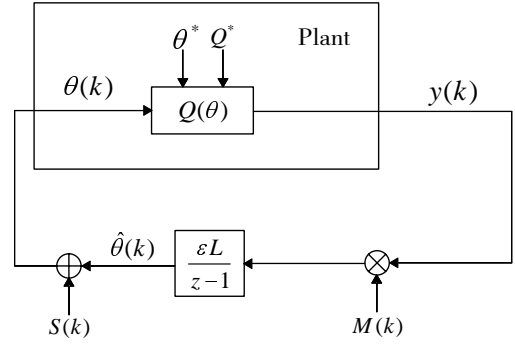}\caption{Extremum seeking
control scheme for discrete-time systems}%
\label{ES2}%
\end{figure}

Define the estimation error $\tilde{\theta}(k)$ as%
\[
\tilde{\theta}(k)=\hat{\theta}(k)-\theta^{\ast}.
\]
Then by (\ref{eq120}), the estimation error is governed by%
\begin{equation}
\left.
\begin{array}
[c]{l}%
\tilde{\theta}\left(  k+1\right)  =\tilde{\theta}\left(  k\right)
+\varepsilon LM(k)\big[Q^{\ast}\\
+\frac{1}{2}(\tilde{\theta}(k)+S(k))^{\mathrm{T}}H(\tilde{\theta
}(k)+S(k))\big]\\
=\tilde{\theta}\left(  k\right)  +\varepsilon LM(k)\big[Q^{\ast}+\frac{1}%
{2}S^{\mathrm{T}}(k)HS(k)\\
+\frac{1}{2}\tilde{\theta}^{\mathrm{T}}(k)H\tilde{\theta}(k)+S^{\mathrm{T}%
}(k)H\tilde{\theta}(k)\big].
\end{array}
\right.  \label{eq112}%
\end{equation}
To analyze the ES control system (\ref{eq112}), the averaging theory based on
\cite{bfs88tcs} was used in the existing literature (see
\cite{ckal02tac,fkb13ejc}). To be specific, let $\omega_{i}=\frac{2\pi
\alpha_{i}}{T},$ $T\in \mathbf{N}_{+},$ $\alpha_{i}\in \mathbf{Z\backslash
\{}0\mathbf{\}}$ ($i\in \mathbf{I[}1,n\mathbf{]}$) satisfying $\left \vert
2\alpha_{i}/T\right \vert <1$ and $\alpha_{i}\neq \alpha_{j},$ $i\neq j.$ This
guarantees that (\ref{eq11}), (\ref{eq12}) and (\ref{eq13}) hold below. The
averaged system of (\ref{eq112}) can be derived as \cite{fkb13ejc}%
\begin{equation}
\tilde{\theta}_{\mathrm{av}}\left(  k+1\right)  =(I_{n}+\varepsilon
LH)\tilde{\theta}_{\mathrm{av}}(k),\label{eq34}%
\end{equation}
which is exponentially stable when $\varepsilon$ is small enough since $LH$ is
Hurwitz. Similar to the continuous-time case, the basic
problem in the averaging method is also to determine in what sense the
behavior of the averaged system (\ref{eq34}) approximates the behavior of the
original system (\ref{eq112}), which may not be intuitively clear. Moreover,
the classical averaging leads to a qualitative analysis.

Inspired by \cite{yzf22tac}, we apply the time-delay to averaging of system
(\ref{eq112}). Summing in $k\geq T-1$ from $k-T+1$ to $k$ and dividing by $T$
on both sides of (\ref{eq112}), we get%
\begin{equation}
\left.
\begin{array}
[c]{l}%
\frac{1}{T}%
{\textstyle \sum \limits_{i=k-T+1}^{k}}
\tilde{\theta}\left(  i+1\right)  \\
=\frac{1}{T}%
{\textstyle \sum \limits_{i=k-T+1}^{k}}
\tilde{\theta}\left(  i\right)  +\frac{\varepsilon}{T}%
{\textstyle \sum \limits_{i=k-T+1}^{k}}
LM(i)Q^{\ast}\\
+\frac{\varepsilon}{2T}%
{\textstyle \sum \limits_{i=k-T+1}^{k}}
LM(i)S^{\mathrm{T}}(i)HS(i)\\
+\frac{\varepsilon}{2T}%
{\textstyle \sum \limits_{i=k-T+1}^{k}}
LM(i)\tilde{\theta}^{\mathrm{T}}(i)H\tilde{\theta}(i)\\
+\frac{\varepsilon}{T}%
{\textstyle \sum \limits_{i=k-T+1}^{k}}
LM(i)S^{\mathrm{T}}(i)H\tilde{\theta}(i),\text{ }k\geq T-1.
\end{array}
\right.  \label{eq113}%
\end{equation}
Set%
\begin{equation}
\bar{\theta}\left(  j\right)  =\tilde{\theta}\left(  j+1\right)
-\tilde{\theta}\left(  j\right)  .\label{eq113a}%
\end{equation}
For the term on the left-hand side of (\ref{eq113}), we have%
\begin{equation}
\left.
\begin{array}
[c]{l}%
\frac{1}{T}%
{\textstyle \sum \limits_{i=k-T+1}^{k}}
\tilde{\theta}\left(  i+1\right)  =\frac{1}{T}%
{\textstyle \sum \limits_{i=k-T+1}^{k}}
[\tilde{\theta}\left(  i+1\right)  \pm \tilde{\theta}\left(  k+1\right)  ]\\
=\frac{1}{T}%
{\textstyle \sum \limits_{i=k-T+1}^{k}}
\tilde{\theta}\left(  k+1\right)  -\frac{1}{T}%
{\textstyle \sum \limits_{i=k-T+1}^{k}}
[\tilde{\theta}\left(  k+1\right)  -\tilde{\theta}\left(  i+1\right)  ]\\
=\frac{1}{T}%
{\textstyle \sum \limits_{i=k-T+1}^{k}}
\tilde{\theta}\left(  k+1\right)  -\frac{1}{T}%
{\textstyle \sum \limits_{i=k-T+2}^{k+1}}
[\tilde{\theta}\left(  k+1\right)  -\tilde{\theta}\left(  i\right)  ]\\
=\tilde{\theta}\left(  k+1\right)  -\frac{1}{T}%
{\textstyle \sum \limits_{i=k-T+2}^{k}}
[\tilde{\theta}\left(  k+1\right)  -\tilde{\theta}\left(  i\right)  ]\\
=\tilde{\theta}\left(  k+1\right)  -\frac{1}{T}%
{\textstyle \sum \limits_{i=k-T+2}^{k}}
\text{ }%
{\textstyle \sum \limits_{j=i}^{k}}
\bar{\theta}\left(  j\right)  .
\end{array}
\right.  \label{eq97}%
\end{equation}
For the first term on the right-hand side of (\ref{eq113}), we have%
\begin{equation}
\left.
\begin{array}
[c]{l}%
\frac{1}{T}%
{\textstyle \sum \limits_{i=k-T+1}^{k}}
\tilde{\theta}\left(  i\right)  =\frac{1}{T}%
{\textstyle \sum \limits_{i=k-T+1}^{k}}
[\tilde{\theta}\left(  i\right)  \pm \tilde{\theta}\left(  k\right)  ]\\
=\frac{1}{T}%
{\textstyle \sum \limits_{i=k-T+1}^{k}}
\tilde{\theta}\left(  k\right)  -\frac{1}{T}%
{\textstyle \sum \limits_{i=k-T+1}^{k}}
[\tilde{\theta}\left(  k\right)  -\tilde{\theta}\left(  i\right)  ]\\
=\tilde{\theta}\left(  k\right)  -\frac{1}{T}%
{\textstyle \sum \limits_{i=k-T+1}^{k-1}}
[\tilde{\theta}\left(  k\right)  -\tilde{\theta}\left(  i\right)  ]\\
=\tilde{\theta}\left(  k\right)  -\frac{1}{T}%
{\textstyle \sum \limits_{i=k-T+1}^{k-1}}
\text{ }%
{\textstyle \sum \limits_{j=i}^{k-1}}
\bar{\theta}\left(  j\right)  .
\end{array}
\right.  \label{eq98}%
\end{equation}
For the second term on the right-hand side of (\ref{eq113}), we have%
\begin{equation}
\left.
\begin{array}
[c]{l}%
\frac{\varepsilon}{T}%
{\textstyle \sum \limits_{i=k-T+1}^{k}}
LM(k)Q^{\ast}\\
=\frac{\varepsilon}{T}Q^{\ast}L\operatorname{col}\Big \{ \frac{2}{a_{j}}%
{\textstyle \sum \limits_{i=k-T+1}^{k}}
\sin \big(\frac{2\pi \alpha_{j}}{T}i\big)\Big \}_{j=1}^{n}\\
=0,
\end{array}
\right.  \label{eq114}%
\end{equation}
where we have utilized%
\begin{equation}
\sum \limits_{i=k-T+1}^{k}\sin \Big(\frac{2\pi \alpha_{j}}{T}i\Big)=0,\text{
\ }j\in \mathbf{I[}1,n\mathbf{]}.\label{eq11}%
\end{equation}
For the third term on the right-hand side of (\ref{eq113}), we have%
\begin{equation}
\left.
\begin{array}
[c]{l}%
\frac{\varepsilon}{2T}%
{\textstyle \sum \limits_{i=k-T+1}^{k}}
LM(i)S^{\mathrm{T}}(i)HS(i)\\
=\frac{\varepsilon L}{2T}%
{\textstyle \sum \limits_{i=k-T+1}^{k}}
\text{ }%
{\textstyle \sum \limits_{j=1}^{n}}
\text{ }%
{\textstyle \sum \limits_{s=1}^{n}}
a_{j}a_{s}h_{js}\sin \big(\frac{2\pi \alpha_{j}}{T}i\big)\sin \big(\frac
{2\pi \alpha_{s}}{T}i\big)M(i)\\
=\frac{\varepsilon L}{T}\operatorname{col}\Big \{%
{\textstyle \sum \limits_{j=1}^{n}}
\text{ }%
{\textstyle \sum \limits_{s=1}^{n}}
\frac{a_{j}a_{s}h_{js}}{a_{m}}%
{\textstyle \sum \limits_{i=k-T+1}^{k}}
\sin \big(\frac{2\pi \alpha_{j}}{T}i\big)\\
\times \sin \big(\frac{2\pi \alpha_{s}}{T}i\big)\sin \big(\frac{2\pi \alpha_{m}}%
{T}i\big)\Big \}_{m=1}^{n}\\
=0,
\end{array}
\right.  \label{eq115}%
\end{equation}
where we have utilized%
\begin{equation}
\left.  \sum \limits_{i=k-T+1}^{k}\sin \Big(\frac{2\pi \alpha_{j}}{T}%
i\Big)\sin \Big(\frac{2\pi \alpha_{s}}{T}i\Big)\sin \Big(\frac{2\pi \alpha_{m}}%
{T}i\Big)=0.\right.  \label{eq12}%
\end{equation}
For the fourth term on the right-hand side of (\ref{eq113}), we have%
\begin{equation}
\left.
\begin{array}
[c]{l}%
\frac{\varepsilon}{2T}%
{\textstyle \sum \limits_{i=k-T+1}^{k}}
LM(i)\tilde{\theta}^{\mathrm{T}}(i)H\tilde{\theta}(i)\\
=\frac{\varepsilon}{2T}%
{\textstyle \sum \limits_{i=k-T+1}^{k}}
LM(i)[\tilde{\theta}^{\mathrm{T}}(i)H\tilde{\theta}(i)\pm \tilde{\theta
}^{\mathrm{T}}(k)H\tilde{\theta}(k)]\\
=\frac{\varepsilon L}{2T}\tilde{\theta}^{\mathrm{T}}(k)H\tilde{\theta}(k)%
{\textstyle \sum \limits_{i=k-T+1}^{k}}
M(i)\\
-\frac{\varepsilon}{2T}%
{\textstyle \sum \limits_{i=k-T+1}^{k-1}}
LM(i)[\tilde{\theta}^{\mathrm{T}}(k)H\tilde{\theta}(k)-\tilde{\theta
}^{\mathrm{T}}(i)H\tilde{\theta}(i)]\\
=-\frac{\varepsilon}{2T}%
{\textstyle \sum \limits_{i=k-T+1}^{k-1}}
LM(i)[\tilde{\theta}^{\mathrm{T}}(k)H\tilde{\theta}(k)-\tilde{\theta
}^{\mathrm{T}}(i)H\tilde{\theta}(i)]\\
=-\frac{\varepsilon}{2T}%
{\textstyle \sum \limits_{i=k-T+1}^{k-1}}
LM(i)[\tilde{\theta}^{\mathrm{T}}(k)H\tilde{\theta}(k)-\tilde{\theta
}^{\mathrm{T}}(k)H\tilde{\theta}(i)\\
+\tilde{\theta}^{\mathrm{T}}(k)H\tilde{\theta}(i)-\tilde{\theta}^{\mathrm{T}%
}(i)H\tilde{\theta}(i)]\\
=-\frac{\varepsilon}{2T}%
{\textstyle \sum \limits_{i=k-T+1}^{k-1}}
LM(i)[\tilde{\theta}^{\mathrm{T}}(k)+\tilde{\theta}^{\mathrm{T}}(i)]H%
{\textstyle \sum \limits_{j=i}^{k-1}}
\bar{\theta}(j)\\
=-\frac{\varepsilon}{2T}%
{\textstyle \sum \limits_{i=k-T+1}^{k-1}}
\text{ }%
{\textstyle \sum \limits_{j=i}^{k-1}}
LM(i)[\tilde{\theta}^{\mathrm{T}}(k)+\tilde{\theta}^{\mathrm{T}}%
(i)]H\bar{\theta}(j),
\end{array}
\right.  \label{eq116}%
\end{equation}
where we have utilized $%
{\textstyle \sum \nolimits_{i=k-T+1}^{k}}
M(i)$ $=0$ via (\ref{eq11}). For the fifth term on the right-hand side of
(\ref{eq113}), we have%
\begin{equation}
\left.
\begin{array}
[c]{l}%
\frac{\varepsilon}{T}%
{\textstyle \sum \limits_{i=k-T+1}^{k}}
LM(i)S^{\mathrm{T}}(i)H\tilde{\theta}(i)\\
=\frac{\varepsilon}{T}%
{\textstyle \sum \limits_{i=k-T+1}^{k}}
LM(i)S^{\mathrm{T}}(i)H[\tilde{\theta}(i)\pm \tilde{\theta}(k)]\\
=\frac{\varepsilon L}{T}%
{\textstyle \sum \limits_{i=k-T+1}^{k}}
M(i)S^{\mathrm{T}}(i)H\tilde{\theta}(k)\\
-\frac{\varepsilon}{T}%
{\textstyle \sum \limits_{i=k-T+1}^{k-1}}
LM(i)S^{\mathrm{T}}(i)H[\tilde{\theta}(k)-\tilde{\theta}(i)]\\
=\varepsilon LH\tilde{\theta}(k)-\frac{\varepsilon}{T}%
{\textstyle \sum \limits_{i=k-T+1}^{k-1}}
\text{ }%
{\textstyle \sum \limits_{j=i}^{k-1}}
LM(i)S^{\mathrm{T}}(i)H\bar{\theta}(j),
\end{array}
\right.  \label{eq117}%
\end{equation}
where we have utilized%
\[
\left.
{\textstyle \sum \limits_{i=k-T+1}^{k}}
M(i)S^{\mathrm{T}}(i)=TI_{n},\right.
\]
since%
\begin{equation}
\left.
{\textstyle \sum \limits_{i=k-T+1}^{k}}
\frac{2a_{j}}{a_{s}}\sin \Big(\frac{2\pi \alpha_{j}}{T}i\Big)\sin \Big(\frac
{2\pi \alpha_{s}}{T}i\Big)=\left \{
\begin{array}
[c]{cc}%
T, & j=s,\\
0, & j\neq s.
\end{array}
\right.  \right.  \label{eq13}%
\end{equation}
Finally, employing (\ref{eq97})-(\ref{eq114}), (\ref{eq115}), (\ref{eq116}%
)-(\ref{eq117}) and setting%
\begin{equation}
\left.
\begin{array}
[c]{l}%
G(k)=\frac{1}{T}%
{\textstyle \sum \limits_{i=k-T+1}^{k-1}}
\text{ }%
{\textstyle \sum \limits_{j=i}^{k-1}}
\bar{\theta}\left(  j\right)  ,\\
Y_{1}(k)=\frac{1}{2T}%
{\textstyle \sum \limits_{i=k-T+1}^{k-1}}
\text{ }%
{\textstyle \sum \limits_{j=i}^{k-1}}
LM(i)[\tilde{\theta}^{\mathrm{T}}(k)+\tilde{\theta}^{\mathrm{T}}%
(i)]H\bar{\theta}(j),\\
Y_{2}(k)=\frac{1}{T}%
{\textstyle \sum \limits_{i=k-T+1}^{k-1}}
\text{ }%
{\textstyle \sum \limits_{j=i}^{k-1}}
LM(i)S^{\mathrm{T}}(i)H\bar{\theta}(j),
\end{array}
\right.  \label{eq121}%
\end{equation}
system (\ref{eq113}) can be transformed to%
\begin{equation}
\left.
\begin{array}
[c]{l}%
\tilde{\theta}\left(  k+1\right)  -G(k+1)=(I_{n}+\varepsilon LH)\tilde{\theta
}\left(  k\right)  \\
-G(k)-\varepsilon Y_{1}(k)-\varepsilon Y_{2}(k),\text{ }k\geq T-1.
\end{array}
\right.  \label{eq118a}%
\end{equation}
System (\ref{eq118a}) is a discrete-time version of the neutral type
time-delay system w.r.t. $\tilde{\theta}.$ The solution $\tilde{\theta}\left(
k\right)  $ of system (\ref{eq112}) is also a solution of the time-delay
system (\ref{eq118a}). Thus, the practical stability of the time-delay system
(\ref{eq118a}) guarantees the practical stability of the original delay-free
ES system (\ref{eq112}). Obviously, we can extend the L-K approach in
\cite{zf22auto} to the discrete-time case to solve the practical stability of
system (\ref{eq118a}). However, the stability analysis will be complicated and the
corresponding results will be more conservative as we shown in the
continuous-time case.

Therefore, for simplifying the stability analysis, we further set%
\begin{equation}
\left.  z(k)=\tilde{\theta}\left(  k\right)  -G(k).\right.  \label{eq118b}%
\end{equation}
Then, system (\ref{eq118a}) can be rewritten as%
\begin{equation}
\left.
\begin{array}
[c]{l}%
z(k+1)=(I_{n}+\varepsilon LH)z(k)\\
+\varepsilon LHG(k)-\varepsilon Y_{1}(k)-\varepsilon Y_{2}(k),\text{ }k\geq
T-1.
\end{array}
\right.  \label{eq118}%
\end{equation}
Comparatively to the averaged system (\ref{eq34}), system (\ref{eq118}) has
the additional terms $G(k),$ $Y_{1}(k)$ and $Y_{2}(k)$ that are all of the
order of \textrm{O}$(\varepsilon)$ provided $\tilde{\theta}(k)$ (and thus $z(k)$) are of the order of $\textrm{O}(1)$. Therefore, for
small $\varepsilon>0$ system (\ref{eq118}) can be regarded as a perturbation
of system (\ref{eq34}). The resulting bound on $|z|$ will lead to the bound on
$\tilde{\theta}:|\tilde{\theta}|\leq|z|+|G|.$ We will find the bound on $z$ by
utilizing the variation of constants formula.

\begin{theorem}
\label{theorem2}Let \textbf{A1-A3} be satisfied. Consider the closed-loop
system (\ref{eq112}) with the initial condition $|\tilde{\theta}(0)|\leq
\sigma_{0}.$ Given tuning parameters $l_{i},$ $a_{i}$ ($i=\mathbf{I[}%
1,n\mathbf{]}$), $\lambda>0$ and $\varepsilon^{\ast}>0$ subject to $\lambda
\varepsilon^{\ast}<1,$ let matrix $P$ ($I_{n}\leq P\leq pI_{n}$) with a scalar
$p\geq1$ and scalar $\zeta$ satisfy the following LMI:%
\begin{equation}
\left.
\begin{array}
[c]{l}%
\Phi_{1}=\left[
\begin{array}
[c]{cc}%
\Phi_{11} & PL+\varepsilon^{\ast}\bar{H}L^{\mathrm{T}}PL\\
\ast & -\zeta I_{n}+\varepsilon^{\ast}L^{\mathrm{T}}PL
\end{array}
\right] <0 ,\\
\Phi_{11}=\bar{H}^{\mathrm{T}}L^{\mathrm{T}}P+PL\bar{H}+\varepsilon^{\ast}%
\bar{H}^{\mathrm{T}}L^{\mathrm{T}}PL\bar{H}\\
\text{ \  \  \  \  \  \ }+2\lambda P+\zeta \kappa^{2}I_{n}.
\end{array}
\right.  \label{eq126a}%
\end{equation}
Given $\sigma>\sigma_{0}>0,$ let the following inequality holds:
\begin{equation}
\left.
\begin{array}
[c]{l}%
\Phi_{2}=p\left(  \sigma_{0}+\frac{\varepsilon^{\ast}\Delta \left[
3(T-1)\lambda+2(\Delta_{1}+\Delta_{2}+\Delta_{3})\right]  }{2\lambda}\right)
^{2}\\
\text{ \  \  \  \ }-\left(  \sigma-\frac{(T-1)\varepsilon^{\ast}\Delta}%
{2}\right)  ^{2}<0,
\end{array}
\right.  \label{eq126}%
\end{equation}
where%
\begin{equation}
\left.
\begin{array}
[c]{l}%
\Delta=\left[  Q_{M}^{\ast}+\frac{H_{M}}{2}\left(  \sigma+\sqrt{%
{\textstyle \sum \nolimits_{i=1}^{n}}
a_{i}^{2}}\right)  ^{2}\right]  \sqrt{%
{\textstyle \sum \nolimits_{i=1}^{n}}
\frac{4l_{i}^{2}}{a_{i}^{2}}},\\
\Delta_{1}=\frac{(T-1)H_{M}\max_{i\in \mathbf{I[}1,n\mathbf{]}}\left \vert
l_{i}\right \vert }{2},\text{ }\Delta_{2}=\frac{(T-1)\sigma H_{M}}{2}\sqrt{%
{\textstyle \sum \nolimits_{i=1}^{n}}
\frac{4l_{i}^{2}}{a_{i}^{2}}},\\
\Delta_{3}=\frac{(T-1)H_{M}}{2}\sqrt{%
{\textstyle \sum \nolimits_{i=1}^{n}}
\frac{4l_{i}^{2}}{a_{i}^{2}}}\sqrt{%
{\textstyle \sum \nolimits_{i=1}^{n}}
a_{i}^{2}}.
\end{array}
\right.  \label{eq41}%
\end{equation}
Then for all $\varepsilon \in(0,\varepsilon^{\ast}],$ the solution of the
closed-loop system (\ref{eq112}) satisfies%
\begin{equation}
\left.
\begin{array}
[c]{l}%
|\tilde{\theta}(k)|<|\tilde{\theta}(0)|+\varepsilon \left(  T-1\right)
\Delta,\text{ }k\in \mathbf{I}[0,T-1],\\
|\tilde{\theta}(k)|<\sqrt{p}(1-\lambda \varepsilon)^{k-T+1}\left(
|\tilde{\theta}(0)|+\frac{3(T-1)\varepsilon \Delta}{2}\right)  \\
+\frac{\varepsilon \Delta(\Delta_{1}+\Delta_{2}+\Delta_{3})\sqrt{p}}{\lambda
}+\frac{(T-1)\varepsilon \Delta}{2}<\sigma,\text{ }k\geq T-1.
\end{array}
\right.  \label{eq133}%
\end{equation}
Moreover, for all $\varepsilon \in(0,\varepsilon^{\ast}]$ and all initial
conditions $|\tilde{\theta}(0)|\leq \sigma_{0},$ the ball%
\[
\left.  \Theta=\left \{  \tilde{\theta}\in \mathbf{R:}|\tilde{\theta
}|<\varepsilon \Delta \left[  \frac{T-1}{2}+\frac{\sqrt{p}(\Delta_{1}+\Delta
_{2}+\Delta_{3})}{\lambda}\right]  \right \}  \right.  \label{eq90}%
\]
is exponential attractive with a decay rate $1-\lambda \varepsilon.$
\end{theorem}

\begin{proof}
See Appendix A2.
\end{proof}

\begin{remark}
Given any $\sigma_{0}$ and $\sigma^{2}>p\sigma_{0}^{2},$ inequality ($\Phi
_{2}<0$ in (\ref{eq126})) is always feasible for small enough $\varepsilon
^{\ast},$ Therefore, the result is semi-global. To the best of our knowledge,
the existing results based on the averaging theory for discrete-time ES are
qualitative (for example, \cite{ckal02tac,fkb13ejc}), i.e., the system is
stable for small $\varepsilon$ if the averaged system is stable. By contrast,
we provide, for the first time, an effective quantitative analysis method for
discrete-time ES, i.e., we can find a quantitative upper bound of
$\varepsilon$ that ensures the practical stability. Moreover, our method can
make the stability analysis very simple and easy to follow.
\end{remark}

\begin{remark}
For simplicity, we let $L=lI_{n}$ with $l<0$ being a given scalar. Then
following the arguments in Remark \ref{remark4}, we find that $\varepsilon
^{\ast}$ is of the order of \textrm{O}$(1/\left \vert l\right \vert ).$ Thus, a
smaller $\left \vert l\right \vert $ leads to a larger $\varepsilon^{\ast}.$
However, we also find that $\lambda$ is of the order of \textrm{O}$(\left \vert
l\right \vert ),$ then the decay rate $1-\lambda \varepsilon^{\ast}$ is of the
order of \textrm{O}$(1),$ which implies that we cannot adjust the value of the
decay rate $1-\lambda \varepsilon^{\ast}$ by changing the gain $L=lI_{n}.$ This
is different from the continuous-time case, and also shows the conservatism. In
addition, for given $\sigma_{0},$ $\varepsilon,$ $\lambda$ as well as $a_{i},$
$l_{i}$ ($i\in \mathbf{I[}1,n\mathbf{]}$), we can find the UB by repeating the
same process with that in Remark \ref{remark4}. Also, the lower bound of UB is
given by (\ref{eq90}) with $\sigma=0.$
\end{remark}

Next we consider a special case that the unknown Hessian $H$ is a diagonal
matrix and satisfies (\ref{eq71}). In this case, we can directly compute that
for all $k\geq0,$%
\[
\left.  \left \Vert (I_{n}+\varepsilon LH)^{k}\right \Vert \leq \left(
1-\varepsilon H_{m}\min_{i\in \mathbf{I[}1,n\mathbf{]}}\left \vert
l_{i}\right \vert \right)  ^{k}\triangleq(1-\lambda \varepsilon)^{k}.\right.
\]
Then following the arguments in Theorem \ref{theorem2}, we can present the
following corollary.

\begin{corollary}
\label{corollary2}Let \textbf{A1-A2} be satisfied and the diagonal Hessian $H$
be unknown but satisfy (\ref{eq71}). Consider the closed-loop system
(\ref{eq112}) with the initial condition $|\tilde{\theta}(0)|\leq \sigma_{0}.$
Given tuning parameters $k_{i},$ $a_{i}$ ($i=\mathbf{I[}1,n\mathbf{]}$),
$\sigma>\sigma_{0}>0,$ $\lambda>0$ and $\varepsilon^{\ast}>0$ subject to
$\lambda \varepsilon^{\ast}<1,$ let the following inequality holds:
\[
\left.  \Phi=\sigma_{0}+\frac{\varepsilon^{\ast}\Delta \lbrack(\Delta
_{1}+\Delta_{2}+\Delta_{3})+2\left(  T-1\right)  \lambda]}{\lambda}%
<\sigma,\right.
\]
where $\Delta$ and $\Delta_{i}$ $(i=\mathbf{I[}1,3\mathbf{]})$ are given by
(\ref{eq41}). Then for all $\varepsilon \in(0,\varepsilon^{\ast}],$ the
solution of the estimation error system (\ref{eq112}) satisfies%
\[
\left.
\begin{array}
[c]{l}%
|\tilde{\theta}(k)|<|\tilde{\theta}(0)|+\varepsilon \left(  T-1\right)
\Delta,\text{ }k\in \mathbf{I}[0,T-1],\\
|\tilde{\theta}(k)|<(1-\lambda \varepsilon)^{k-T+1}\left(  |\tilde{\theta
}(0)|+\frac{3(T-1)\varepsilon \Delta}{2}\right)  \\
+\frac{\varepsilon \Delta \left[  2(\Delta_{1}+\Delta_{2}+\Delta_{3}%
)+(T-1)\lambda \right]  }{2\lambda}<\sigma,\text{ }k\geq T-1.
\end{array}
\right.
\]
Moreover, for all $\varepsilon \in(0,\varepsilon^{\ast}]$ and all initial
conditions $ \vert \tilde{\theta}(0) \vert \leq \sigma_{0},$ the ball%
\[
\left.  \Theta=\left \{  \tilde{\theta}\in \mathbf{R:}|\tilde{\theta}%
(k)|<\frac{\varepsilon \Delta \left[  2(\Delta_{1}+\Delta_{2}+\Delta
_{3})+(T-1)\lambda \right]  }{2\lambda}\right \}  \right.
\]
is exponential attractive with a decay rate $1-\lambda \varepsilon$ with
$\lambda=H_{m}\min_{i\in \mathbf{I[}1,n\mathbf{]}}\left \vert l_{i}\right \vert
.$
\end{corollary}

\begin{remark}
When $n=1,$ Corollary \ref{corollary2} can be further improved as follows.
Note that%
\[
\left.
\begin{array}
[c]{l}%
\left \vert G(k)\right \vert <\frac{(T-1)\varepsilon \Delta}{2},\text{
}\left \vert LHG(k)\right \vert <\frac{(T-1)\varepsilon \Delta \cdot \left \vert
LH\right \vert }{2},\\
|Y_{1}(k)|<\frac{(T-1)\varepsilon \Delta \cdot \sigma \left \vert LH\right \vert
}{\left \vert a\right \vert },\text{ }|Y_{2}(k)|<(T-1)\varepsilon \Delta
\cdot \left \vert LH\right \vert ,\\
\left \vert (1+\varepsilon LH)^{k}\right \vert \leq \left(  1-\varepsilon
LH_{m}\right)  ^{k}\triangleq(1-\lambda \varepsilon)^{k},
\end{array}
\right.
\]
where%
\[
\left.  \Delta=\left[  Q_{M}^{\ast}+\frac{H_{M}}{2}\left(  \sigma+\left \vert
a\right \vert \right)  ^{2}\right]  \frac{2\left \vert L\right \vert }{\left \vert
a\right \vert }.\right.
\]
Then via (\ref{eq122}), we have%
\[
\left.
\begin{array}
[c]{l}%
|z(k)|<(1-\lambda \varepsilon)^{(k-T+1)}|z(T-1)|\\
+\frac{\varepsilon^{2}\Delta \cdot \left(  T-1\right)  \left(  3a+2\sigma
\right)  \left \vert LH\right \vert }{2a}%
{\textstyle \sum \limits_{i=T-1}^{k-1}}
(1-\left \vert LH\right \vert \varepsilon)^{(k-i-1)}\\
=(1-\lambda \varepsilon)^{(k-T+1)}|z(T-1)|\\
+\frac{\varepsilon \Delta \cdot \left(  T-1\right)  \left(  3a+2\sigma \right)
}{2a}\left[  1-(1-\left \vert LH\right \vert \varepsilon)^{(k-T+1)}\right]  \\
\leq(1-\lambda \varepsilon)^{(k-T+1)}|z(T-1)|+\frac{\varepsilon \Delta
\cdot \left(  T-1\right)  \left(  3a+2\sigma \right)  }{2a}%
\end{array}
\right.
\]
with $k\geq T-1.$ It follows that%
\[
\left.
\begin{array}
[c]{l}%
|\tilde{\theta}(k)|=(1-\lambda \varepsilon)^{k-T+1}\left(  |\tilde{\theta
}(0)|+\frac{3(T-1)\varepsilon \Delta}{2}\right)  \\
+\frac{\varepsilon \Delta \cdot \left(  T-1\right)  \left(  2a+\sigma \right)
}{a},\text{ }%
\end{array}
\right.
\]
which implies that%
\[
\left.
\begin{array}
[c]{l}%
|\tilde{\theta}(k)|=(1-\lambda \varepsilon)^{k-T+1}\left(  |\tilde{\theta
}(0)|+\frac{3(T-1)\varepsilon \Delta}{2}\right)  \\
+\frac{\varepsilon \Delta \cdot \left(  T-1\right)  \left(  2a+\sigma \right)
}{a}<\sigma,\text{ }k\geq T-1,
\end{array}
\right.
\]
if%
\[
\left.  \Phi=\sigma_{0}+\frac{\varepsilon^{\ast}\Delta \cdot \left(  T-1\right)
\left(  7a+2\sigma \right)  }{2a}<\sigma,\text{ }\forall \varepsilon
\in(0,\varepsilon^{\ast}].\right.
\]
In this case, the ultimate bound is given by%
\[
\left.  \Theta=\left \{  \tilde{\theta}\in \mathbf{R:}|\tilde{\theta}%
(k)|<\frac{\varepsilon \Delta \cdot \left(  T-1\right)  \left(  2a+\sigma \right)
}{a}\right \}  \right.
\]
with a decay rate $1-\lambda \varepsilon$ with $\lambda=\left \vert
LH_{m}\right \vert .$
\end{remark}

\subsection{Examples}

\subsubsection{Scalar systems}

Given the single-input map%
\[
Q(\theta(k))=Q^{\ast}+\frac{H}{2}\theta^{2}(k)
\]
with $Q^{\ast}=0$ and $H=2,$ we select the tuning parameters of the
gradient-based ES as $L=-0.1,$ $a=0.2$ and $T=2.$ If $Q^{\ast}$ and $H$ are
unknown, but satisfy \textbf{A2} and (\ref{eq71}) we consider
\[
\left.  Q_{M}^{\ast}=0.5,\text{ }H_{m}=1,\text{ }H_{M}=3.\right.
\]
Both the solutions of uncertainty-free and uncertainty cases are shown in
Table \ref{tab6}.

For the numerical simulations, we choose $\omega=\pi/2$ and the same other
parameter values as shown above. In addition, in the uncertainty case, we let
\[
\left.  H=2+\sin k\right.  ,
\]
which satisfies the condition $1\leq \left \vert H\right \vert \leq3$ as shown in
Table \ref{tab6}. Under the initial condition $\hat{\theta}(0)=1$ for both
cases,\ the simulation results are shown in Fig. \ref{d1c} and Fig. \ref{d1u},
respectively, from which we can see that the values of UB shown in Table
\ref{tab6} are confirmed.

\begin{table}[h]
\caption{Scalar systems}%
\label{tab6}
\center
\setlength{\tabcolsep}{3pt}
\begin{tabular}
[c]{|l|c|c|c|c|c|}\hline
ES: sine wave & $\sigma_{0}$ & $\sigma$ & $\lambda$ & $\varepsilon^{\ast}$ &
UB\\ \hline
$Q^{\ast}=0,H=2$ & 1 & $\sqrt{2}$ & 0.2 & 0.015 &
$1.6\mathrm{e}^{-3}$\\ \hline
$\big \vert Q^{\ast}\big \vert \leq1, 1\leq \big \vert H\big \vert \leq3$ & 1 &
$\sqrt{2}$ & 0.1 & 0.008 & $1.0\mathrm{e}^{-2}$\\ \hline
\end{tabular}
\end{table}

\begin{figure}[ptb]
\centering
\includegraphics[width=0.45\textwidth]{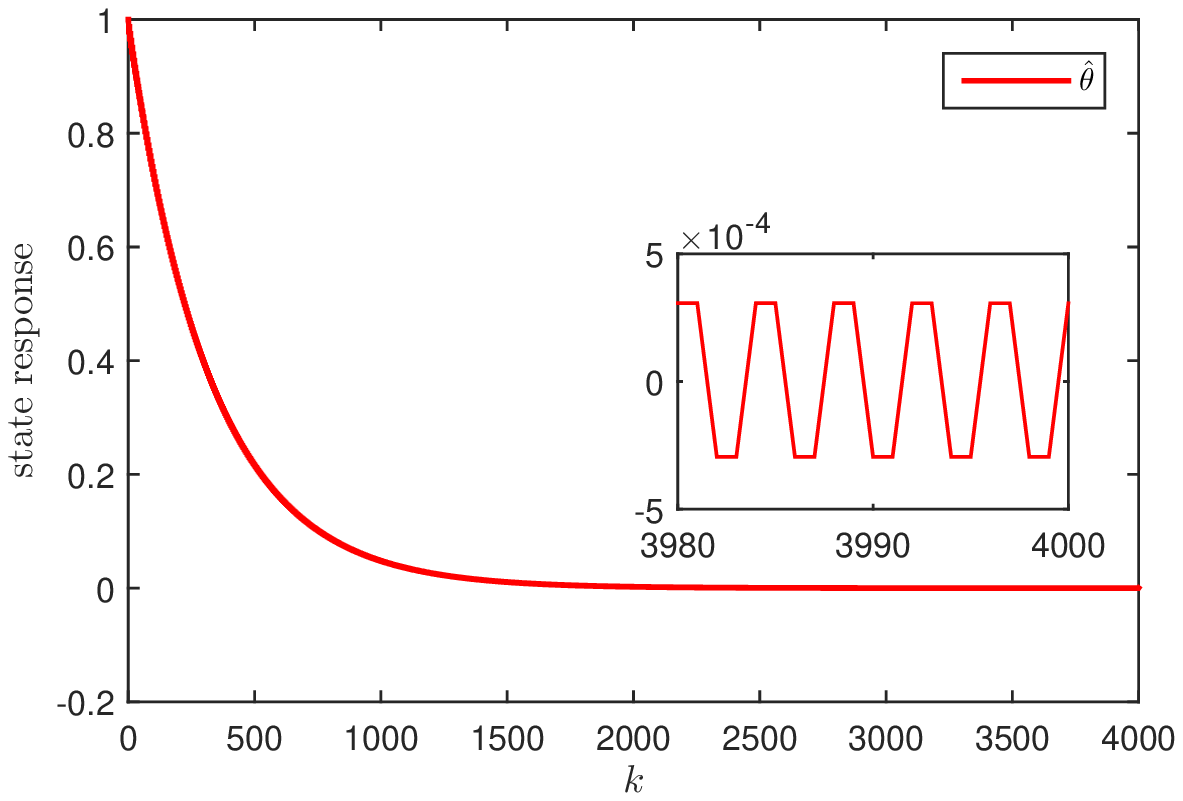}\caption{Trajectory of the
real-time estimate $\hat{\theta}$ for the uncertainty-free case}%
\label{d1c}%
\end{figure}

\begin{figure}[ptb]
\centering
\includegraphics[width=0.45\textwidth]{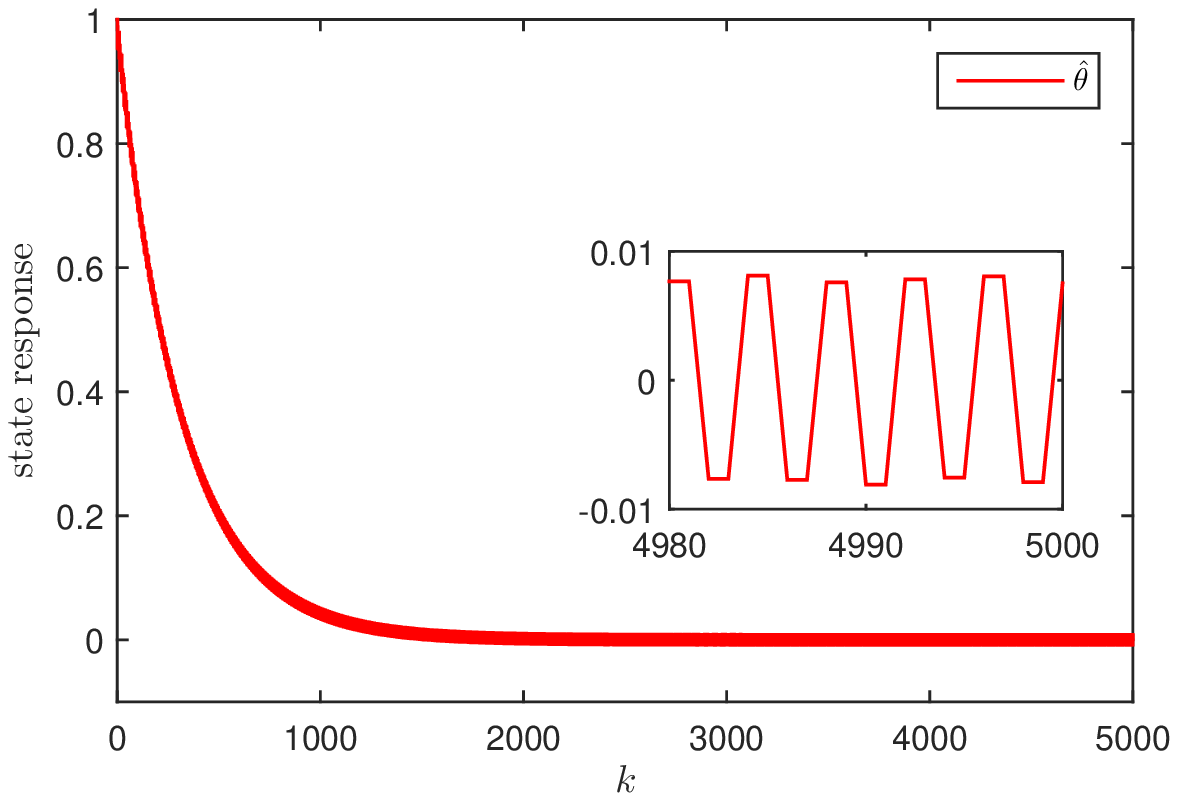}\caption{Trajectory of the
real-time estimate $\hat{\theta}$ for the uncertainty case}%
\label{d1u}%
\end{figure}

\subsubsection{Vector systems}

Consider the quadratic function (\ref{eq119}) with \cite{guay14jpc}%
\[
\left.  Q^{\ast}=1,\text{ }H=\left[
\begin{array}
[c]{cc}%
100 & 30\\
30 & 20
\end{array}
\right]  ,\text{ }\theta^{\ast}=\left[
\begin{array}
[c]{c}%
2\\
4
\end{array}
\right]  .\right.
\]
We select the tuning parameters of the gradient-based ES as $l_{1}%
=l_{2}=-0.001,$ $a=0.5$ and $T=2.$ The solutions are shown in Table \ref{tab8}.

\begin{table}[h]
\caption{Vector systems}%
\label{tab8}%
\center
\setlength{\tabcolsep}{3pt}
\begin{tabular}
[c]{|l|c|c|c|c|c|}\hline
ES: sine wave & $\sigma_{0}$ & $\sigma$ & $\lambda$ & $\varepsilon^{\ast}$ &
UB\\ \hline
Corollary 2 & 1 & $\sqrt{2}$ & 0.11 & 0.034 & $1.96\mathrm{e}^{-2}$\\ \hline
\end{tabular}
\end{table}

For the numerical simulations, we choose the same parameter values as shown
above with $\varepsilon=\varepsilon^{\ast}=0.034$ and $\omega_{2}=-\omega
_{1}=-2\pi/3$. Under the initial condition $\hat{\theta}(0)=[3,3]^{\mathrm{T}%
}$ (thus $\tilde{\theta}(0)=\hat{\theta}(0)-\theta^{\ast}=[1,-1]^{\mathrm{T}}%
$)$,$ the simulation results are shown in Fig. \ref{d2c}. It follows that the
UB value given by Table \ref{tab8} is confirmed.

\begin{figure}[ptb]
\centering
\includegraphics[width=0.45\textwidth]{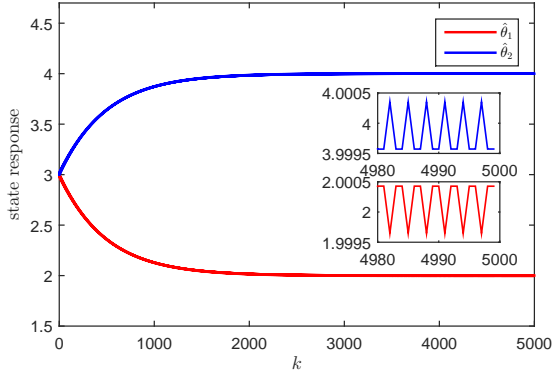}\caption{Trajectories of the
real-time estimate $\hat{\theta}$}%
\label{d2c}%
\end{figure}

\section{conclusion} \label{sec4}

This article developed a time-delay approach to ES both in the continuous and, for the first time, the discrete domains. Significantly simpler and more efficient stability analysis in terms of simple inequalities has been suggested. Explicit
conditions in terms of inequalities were established to guarantee the practical stability of the ES control systems by employing the variation of constants formula to the perturbed averaged system. Comparatively to the L-K method, the established method not
only greatly simplifies the stability analysis, but also improves the results, for instance, allows us to get larger decay rate, period of the dither signal and uncertainties of the map. We finally mention that the proposed method can be applied in the
future to ES where the static maps have sampled-data and delayed measurements. Other possible topics are ES for dynamic maps
and non-quadratic maps.

\section*{Appendix}

\subsection*{A1: Proof of Theorem \ref{theorem1}}

Assume that%
\begin{equation}
\left.  |\tilde{\theta}(t)|<\sigma,\text{ }\forall t\geq0.\right.
\label{eq61}%
\end{equation}
Note from (\ref{eq53})-(\ref{eq17}) and (\ref{eq61}) that%
\begin{equation}
\left.
\begin{array}
[c]{l}%
|y(t)|=\left \vert Q^{\ast}+\frac{1}{2}(\tilde{\theta}(t)+S(t))^{\mathrm{T}%
}H(\tilde{\theta}(t)+S(t))\right \vert \\
<Q_{M}^{\ast}+\frac{H_{M}}{2}\left(  \sigma+\sqrt{%
{\textstyle \sum \nolimits_{i=1}^{n}}
a_{i}^{2}}\right)  ^{2},\text{ }t\geq0,\\
\left \vert \dot{\tilde{\theta}}(t)\right \vert =|KM(t)y(t)|<\Delta,\text{
}t\geq0,\\
|\tilde{\theta}(t)|=\left \vert \tilde{\theta}(0)+%
{\textstyle \int \nolimits_{0}^{t}}
\dot{\tilde{\theta}}(s)\mathrm{d}s\right \vert <|\tilde{\theta}(0)|+\varepsilon
\Delta,\text{ }t\in \lbrack0,\varepsilon]
\end{array}
\right.  \label{eq20}%
\end{equation}
with $\Delta$ given by (\ref{eq70}). The first inequality in (\ref{eq60}) follows from the third inequality in
(\ref{eq20}) since $\Phi_{2}<0$ in (\ref{eq24}) implies that $\sigma
_{0}+\varepsilon^{\ast}\Delta<\sigma,$ $\forall \varepsilon \in(0,\varepsilon
^{\ast}]$. Next we consider the case with $t\geq \varepsilon.$

To make the second inequality in (\ref{eq60}) hold, we use the variation of
constants formula for (\ref{eq52}) to obtain
\begin{equation}
\left.
\begin{array}
[c]{l}%
z(t)=\mathrm{e}^{KH(t-\varepsilon)}z(\varepsilon)\\
+%
{\textstyle \int \nolimits_{\varepsilon}^{t}}
\mathrm{e}^{KH(t-s)}[KHG(s)-Y_{1}(s)-Y_{2}(s)]\mathrm{d}s,\text{ }%
t\geq \varepsilon.
\end{array}
\right.  \label{eq22}%
\end{equation}

From (\ref{eq37}) and (\ref{eq20}) we have%
\begin{equation}
\left.
\begin{array}
[c]{l}%
|G(t)|=\left \vert \frac{1}{\varepsilon}%
{\textstyle \int \nolimits_{t-\varepsilon}^{t}}
(\tau-t+\varepsilon)\dot{\tilde{\theta}}(\tau)\mathrm{d}\tau \right \vert \\
\leq \frac{1}{\varepsilon}%
{\textstyle \int \nolimits_{t-\varepsilon}^{t}}
\left \vert (\tau-t+\varepsilon)\dot{\tilde{\theta}}(\tau)\right \vert
\mathrm{d}\tau \\
<\frac{1}{\varepsilon}\Delta%
{\textstyle \int \nolimits_{t-\varepsilon}^{t}}
(\tau-t+\varepsilon)\mathrm{d}\tau \\
=\frac{\varepsilon \Delta}{2},
\end{array}
\right.  \label{eq21}%
\end{equation}
and%
\begin{equation}
\left.
\begin{array}
[c]{l}%
|KHG(t)|\leq \left \Vert K\right \Vert \left \Vert H\right \Vert |G(t)|\\
<\frac{\varepsilon \Delta \cdot H_{M}\max_{i\in \mathbf{I[}1,n\mathbf{]}%
}\left \vert k_{i}\right \vert }{2}\\
=\varepsilon \Delta \cdot \Delta_{1}%
\end{array}
\right.  \label{eq57}%
\end{equation}
with $\Delta_{1}$ given by (\ref{eq70}). From (\ref{eq56}) and (\ref{eq20}) we
have
\begin{equation}
\left.
\begin{array}
[c]{l}%
|Y_{1}(t)|=\left \vert \frac{1}{\varepsilon}%
{\textstyle \int \nolimits_{t-\varepsilon}^{t}}
{\textstyle \int \nolimits_{\tau}^{t}}
KM(\tau)\tilde{\theta}^{\mathrm{T}}(s)H\dot{\tilde{\theta}}(s)\mathrm{d}%
s\mathrm{d}\tau \right \vert \\
\leq \frac{1}{\varepsilon}%
{\textstyle \int \nolimits_{t-\varepsilon}^{t}}
{\textstyle \int \nolimits_{\tau}^{t}}
\left \vert KM(\tau)\right \vert \left \vert \tilde{\theta}^{\mathrm{T}%
}(s)\right \vert \left \Vert H\right \Vert \left \vert \dot{\tilde{\theta}%
}(s)\right \vert \mathrm{d}s\mathrm{d}\tau \\
<\frac{1}{\varepsilon}%
{\textstyle \int \nolimits_{t-\varepsilon}^{t}}
{\textstyle \int \nolimits_{\tau}^{t}}
\sqrt{%
{\textstyle \sum \nolimits_{i=1}^{n}}
\frac{4k_{i}^{2}}{a_{i}^{2}}}\sigma H_{M}\Delta \mathrm{d}s\mathrm{d}\tau \\
=\frac{\varepsilon}{2}\sqrt{%
{\textstyle \sum \nolimits_{i=1}^{n}}
\frac{4k_{i}^{2}}{a_{i}^{2}}}\sigma H_{M}\Delta=\varepsilon \Delta \cdot
\Delta_{2},
\end{array}
\right.  \label{eq58}%
\end{equation}
and%
\begin{equation}
\left.
\begin{array}
[c]{l}%
|Y_{2}(t)|=\left \vert \frac{1}{\varepsilon}%
{\textstyle \int \nolimits_{t-\varepsilon}^{t}}
{\textstyle \int \nolimits_{\tau}^{t}}
KM(\tau)S^{\mathrm{T}}(\tau)H\dot{\tilde{\theta}}(s)\mathrm{d}s\mathrm{d}%
\tau \right \vert \\
\leq \frac{1}{\varepsilon}%
{\textstyle \int \nolimits_{t-\varepsilon}^{t}}
{\textstyle \int \nolimits_{\tau}^{t}}
\left \vert KM(\tau)\right \vert \left \vert S^{\mathrm{T}}(\tau)\right \vert
\left \Vert H\right \Vert \left \vert \dot{\tilde{\theta}}(s)\right \vert
\mathrm{d}s\mathrm{d}\tau \\
<\frac{1}{\varepsilon}%
{\textstyle \int \nolimits_{t-\varepsilon}^{t}}
{\textstyle \int \nolimits_{\tau}^{t}}
\sqrt{%
{\textstyle \sum \nolimits_{i=1}^{n}}
\frac{4k_{i}^{2}}{a_{i}^{2}}}\sqrt{%
{\textstyle \sum \nolimits_{i=1}^{n}}
a_{i}^{2}}H_{M}\Delta \mathrm{d}s\mathrm{d}\tau \\
=\frac{\varepsilon}{2}\sqrt{%
{\textstyle \sum \nolimits_{i=1}^{n}}
\frac{4k_{i}^{2}}{a_{i}^{2}}}\sqrt{%
{\textstyle \sum \nolimits_{i=1}^{n}}
a_{i}^{2}}H_{M}\Delta=\varepsilon \Delta \cdot \Delta_{3},
\end{array}
\right.  \label{eq59}%
\end{equation}
where $\Delta_{2}$ and $\Delta_{3}$ are given by (\ref{eq70}).
Via (\ref{eq22}) and (\ref{eq57})-(\ref{eq59}), we obtain%
\begin{equation}
\left.
\begin{array}
[c]{l}%
|z(t)|\leq \left \Vert \mathrm{e}^{KH(t-\varepsilon)}\right \Vert |z(\varepsilon
)|\\
+%
{\textstyle \int \nolimits_{\varepsilon}^{t}}
\left \Vert \mathrm{e}^{KH(t-s)}\right \Vert [|KHG(s)|+|Y_{1}(s)|+|Y_{2}%
(s)|]\mathrm{d}s\\
<\left \Vert \mathrm{e}^{KH(t-\varepsilon)}\right \Vert |z(\varepsilon)|\\
+\varepsilon \Delta(\Delta_{1}+\Delta_{2}+\Delta_{3})%
{\textstyle \int \nolimits_{\varepsilon}^{t}}
\left \Vert \mathrm{e}^{KH(t-s)}\right \Vert \mathrm{d}s,\text{ }t\geq
\varepsilon.
\end{array}
\right.  \label{eq23}%
\end{equation}

In order to derive a bound on $\mathrm{e}^{KHt},$ consider the nominal system%
\begin{equation}
\left.  \dot{z}(t)=KHz(t)=K(\bar{H}+\Delta H)z(t),\text{ }t\geq0,\right.
\label{eq26}%
\end{equation}
where we noted \textbf{A3}. Choose the Lyapunov function $V(t)=z^{\mathrm{T}%
}(t)Pz(t)$ with $P$ satisfying $I_{n}\leq P\leq pI_{n}$. Then%
\begin{equation}
\left.
\begin{array}
[c]{l}%
\dot{V}(t)+2\delta V(t)=2z^{\mathrm{T}}(t)P[K(\bar{H}+\Delta H)]z(t)\\
+2\delta z^{\mathrm{T}}(t)Pz(t).
\end{array}
\right.  \label{eq83}%
\end{equation}
To compensate $\Delta Hz(t)$ in (\ref{eq83}) we apply $S$-procedure, we add to
$\dot{V}(t)+2\delta V(t)$ the left hand part of
\[
\left.  \zeta \left(  \kappa^{2}\left \vert z(t)\right \vert ^{2}-\left \vert
\Delta Hz(t)\right \vert ^{2}\right)  \geq0\right.
\]
with some $\zeta>0.$ Then, we have%
\[
\left.
\begin{array}
[c]{l}%
\dot{V}(t)+2\delta V(t)\leq2z^{\mathrm{T}}(t)P[K(\bar{H}+\Delta H)]z(t)\\
+2\delta z^{\mathrm{T}}(t)Pz(t)+\zeta \left(  \kappa^{2}\left \vert
z(t)\right \vert ^{2}-\left \vert \Delta Hz(t)\right \vert ^{2}\right)  \\
=\xi^{\mathrm{T}}(t)\Phi_{1}\xi(t),
\end{array}
\right.
\]
where $\xi^{\mathrm{T}}(t)=[z^{\mathrm{T}}(t),z^{\mathrm{T}}(t)(\Delta
H)^{\mathrm{T}}]$ and $\Phi_{1}$ is given by (\ref{eq24}). Thus, if $\Phi
_{1}<0$ in (\ref{eq24}), we have%
\begin{equation}
\left.  \dot{V}(t)\leq-2\delta V(t),\text{ }t\geq0,\right.  \label{eq84}%
\end{equation}
which with $I_{n}\leq P\leq pI_{n}$ yields%
\[
\left.  |z(t)|^{2}\leq V(t)\leq \mathrm{e}^{-2\delta t}V(0)\leq p\mathrm{e}%
^{-2\delta t}|z(0)|^{2},\right.
\]
namely,%
\begin{equation}
\left.  |z(t)|\leq \sqrt{p}\mathrm{e}^{-\delta t}|z(0)|,\text{ }t\geq0.\right.
\label{eq27}%
\end{equation}
On the other hand, by using the variation of constants formula for
(\ref{eq26}), we have%
\begin{equation}
\left.  z(t)=\mathrm{e}^{KHt}z(0),\text{ }t\geq0.\right.  \label{eq28}%
\end{equation}
By norm's definition and (\ref{eq27})-(\ref{eq28}), we obtain%
\begin{equation}
\left.
\begin{array}
[c]{l}%
\left \Vert \mathrm{e}^{KHt}\right \Vert =\sup_{|z(0)|=1}\left \vert
\mathrm{e}^{KHt}z(0)\right \vert \\
\overset{(\ref{eq28})}{=}\sup_{|z(0)|=1}\left \vert z(t)\right \vert \\
\overset{(\ref{eq27})}{\leq}\sqrt{p}\mathrm{e}^{-\delta t}.
\end{array}
\right.  \label{eq29}%
\end{equation}
With (\ref{eq29}), inequality (\ref{eq23}) can be continued as%
\begin{equation}
\left.
\begin{array}
[c]{l}%
|z(t)|<\sqrt{p}\mathrm{e}^{-\delta(t-\varepsilon)}|z(\varepsilon)|\\
+\varepsilon \Delta(\Delta_{1}+\Delta_{2}+\Delta_{3})\sqrt{p}%
{\textstyle \int \nolimits_{\varepsilon}^{t}}
\mathrm{e}^{-\delta(t-s)}\mathrm{d}s\\
=\sqrt{p}\mathrm{e}^{-\delta(t-\varepsilon)}|z(\varepsilon)|\\
+\frac{\varepsilon \Delta(\Delta_{1}+\Delta_{2}+\Delta_{3})\sqrt{p}}{\delta
}\left(  1-\mathrm{e}^{-\delta(t-\varepsilon)}\right)  \\
\leq \sqrt{p}\mathrm{e}^{-\delta(t-\varepsilon)}|z(\varepsilon)|+\frac
{\varepsilon \Delta(\Delta_{1}+\Delta_{2}+\Delta_{3})\sqrt{p}}{\delta}.
\end{array}
\right.  \label{eq25}%
\end{equation}
Note from (\ref{eq52a}), (\ref{eq20}) and (\ref{eq21}) that%
\[
\left.
\begin{array}
[c]{l}%
|z(\varepsilon)|=|\tilde{\theta}(\varepsilon)-G(\varepsilon)|\leq
|\tilde{\theta}(\varepsilon)|+|G(\varepsilon)|\\
<|\tilde{\theta}(0)|+\varepsilon \Delta+\frac{\varepsilon \Delta}{2}%
=|\tilde{\theta}(0)|+\frac{3\varepsilon \Delta}{2},
\end{array}
\right.
\]
by which, inequality (\ref{eq25}) can be continued as%
\[
\left.
\begin{array}
[c]{l}%
|z(t)|<\sqrt{p}\mathrm{e}^{-\delta(t-\varepsilon)}\big(|\tilde{\theta
}(0)|+\frac{3\varepsilon \Delta}{2}\big)\\
+\frac{\varepsilon \Delta(\Delta_{1}+\Delta_{2}+\Delta_{3})\sqrt{p}}{\delta
},\text{ }t\geq \varepsilon.
\end{array}
\right.
\]
Then%
\[
\left.
\begin{array}
[c]{l}%
\left \vert \tilde{\theta}(t)\right \vert =|z(t)+G(t)|\leq|z(t)|+|G(t)|\\
<\sqrt{p}\mathrm{e}^{-\delta(t-\varepsilon)}\big(|\tilde{\theta}%
(0)|+\frac{3\varepsilon \Delta}{2}\big)\\
+\frac{\varepsilon \Delta(\Delta_{1}+\Delta_{2}+\Delta_{3})\sqrt{p}}{\delta
}+\frac{\varepsilon \Delta}{2},\text{ }t\geq \varepsilon,
\end{array}
\right.
\]
which implies the second inequality in (\ref{eq60}) due to%
\[
\left.  \sqrt{p}\left(  \sigma_{0}+\frac{\varepsilon^{\ast}\Delta
\lbrack2(\Delta_{1}+\Delta_{2}+\Delta_{3})+3\delta]}{2\delta}\right)
+\frac{\varepsilon^{\ast}\Delta}{2}<\sigma,\right.
\]
namely,%
\[
\left.  \sqrt{p}\left(  \sigma_{0}+\frac{\varepsilon^{\ast}\Delta
\lbrack2(\Delta_{1}+\Delta_{2}+\Delta_{3})+3\delta]}{2\delta}\right)
<\sigma-\frac{\varepsilon^{\ast}\Delta}{2}.\right.
\]
The latter, by squaring of both sides, is equivalent to $\Phi_{2}<0$ in
(\ref{eq24}).

By contradiction-based arguments in \cite{zf22auto} (see Appendix A), it can
be proved that (\ref{eq24}) results in (\ref{eq61}). The proof is finished.

\subsection*{A2: Proof of Theorem \ref{theorem2}}

Assume that%
\begin{equation}
\left.  |\tilde{\theta}(k)|<\sigma,\text{ }\forall k\geq0.\right.
\label{eq137}%
\end{equation}
Then follows from (\ref{eq71}), (\ref{eq119})-(\ref{eq112}), (\ref{eq113a})
and (\ref{eq137}) we have%
\begin{equation}
\left.
\begin{array}
[c]{l}%
|y(k)|=\left \vert Q^{\ast}+\frac{1}{2}(\tilde{\theta}(k)+S(k))^{\mathrm{T}%
}H(\tilde{\theta}(k)+S(k))\right \vert \\
<Q_{M}^{\ast}+\frac{H_{M}}{2}\left(  \sigma+\sqrt{%
{\textstyle \sum \nolimits_{i=1}^{n}}
a_{i}^{2}}\right)  ^{2},\text{ }k\geq0,\\
\left \vert \bar{\theta}(k)\right \vert =|\varepsilon LM(k)y(k)|<\varepsilon
\Delta,\text{ }k\geq0,\\
\left \vert \tilde{\theta}(k)\right \vert =\left \vert \tilde{\theta}(0)+%
{\textstyle \sum \nolimits_{i=0}^{k-1}}
\bar{\theta}\left(  i\right)  \right \vert \\
<|\tilde{\theta}(0)|+\left(  T-1\right)  \varepsilon \Delta,\text{ }%
k\in \mathbf{I}[0,T-1]
\end{array}
\right.  \label{eq136}%
\end{equation}
with $\Delta$ given by (\ref{eq41}). The first inequality in (\ref{eq133})
follows from the third inequality in (\ref{eq136}) since $\Phi_{2}<\sigma$ in
(\ref{eq126}) implies that $\sigma_{0}+\varepsilon^{\ast}(T-1)\Delta<\sigma,$
$\forall \varepsilon \in(0,\varepsilon^{\ast}].$

To make the second inequality in (\ref{eq133}) hold, we use the variation of
constants formula for (\ref{eq118}) to get%
\begin{equation}
\left.
\begin{array}
[c]{l}%
z(k)=(I_{n}+\varepsilon LH)^{k-T+1}z(T-1)\\
+\varepsilon%
{\textstyle \sum \limits_{i=T-1}^{k-1}}
(I_{n}+\varepsilon LH)^{k-i-1}[LHG(i)-Y_{1}(i)-Y_{2}(i)]
\end{array}
\right.  \label{eq122}%
\end{equation}
with $k\geq T-1.$ From (\ref{eq121}) and (\ref{eq136}) we have%
\begin{equation}
\left.
\begin{array}
[c]{l}%
|G(k)|=\frac{1}{T}\left \vert
{\textstyle \sum \limits_{i=k-T+1}^{k-1}}
\text{ }%
{\textstyle \sum \limits_{j=i}^{k-1}}
\bar{\theta}\left(  j\right)  \right \vert \\
\leq \frac{1}{T}%
{\textstyle \sum \limits_{i=k-T+1}^{k-1}}
\text{ }%
{\textstyle \sum \limits_{j=i}^{k-1}}
\left \vert \bar{\theta}\left(  j\right)  \right \vert \\
<\frac{(T-1)\varepsilon \Delta}{2},\text{ }k\geq T-1,
\end{array}
\right.  \label{eq36}%
\end{equation}%
\begin{equation}
\left.
\begin{array}
[c]{l}%
\left \vert LHG(k)\right \vert \leq \left \Vert L\right \Vert \left \Vert
H\right \Vert \left \vert G(k)\right \vert \\
<\max_{i\in \mathbf{I[}1,n\mathbf{]}}\left \vert l_{i}\right \vert H_{M}%
\frac{(T-1)\varepsilon \Delta}{2}\\
=\varepsilon \Delta \cdot \Delta_{1},\text{ }k\geq T-1,
\end{array}
\right.  \label{eq123}%
\end{equation}%
\begin{equation}
\left.
\begin{array}
[c]{l}%
|Y_{1}(k)|=\frac{1}{2T}\left \vert
{\textstyle \sum \limits_{i=k-T+1}^{k-1}}
\text{ }%
{\textstyle \sum \limits_{j=i}^{k-1}}
LM(i)[\tilde{\theta}^{\mathrm{T}}(k)+\tilde{\theta}^{\mathrm{T}}%
(i)]H\bar{\theta}(j)\right \vert \\
\leq \frac{1}{2T}%
{\textstyle \sum \limits_{i=k-T+1}^{k-1}}
\text{ }%
{\textstyle \sum \limits_{j=i}^{k-1}}
\left \vert LM(i)\right \vert \left \vert \tilde{\theta}^{\mathrm{T}}%
(k)+\tilde{\theta}^{\mathrm{T}}(i)\right \vert \left \Vert H\right \Vert
\left \vert \bar{\theta}(j)\right \vert \\
<\frac{1}{2T}%
{\textstyle \sum \limits_{i=k-T+1}^{k-1}}
\text{ }%
{\textstyle \sum \limits_{j=i}^{k-1}}
\sqrt{%
{\textstyle \sum \nolimits_{i=1}^{n}}
\frac{4l_{i}^{2}}{a_{i}^{2}}}\cdot2\sigma \cdot H_{M}\cdot \varepsilon \Delta \\
=\frac{\varepsilon \Delta \cdot \sigma H_{M}}{T}\sqrt{%
{\textstyle \sum \nolimits_{i=1}^{n}}
\frac{4l_{i}^{2}}{a_{i}^{2}}}\frac{T(T-1)}{2}\\
=\frac{\varepsilon \Delta \cdot(T-1)\sigma H_{M}}{2}\sqrt{%
{\textstyle \sum \nolimits_{i=1}^{n}}
\frac{4l_{i}^{2}}{a_{i}^{2}}}\\
=\varepsilon \Delta \cdot \Delta_{2},
\end{array}
\right.  \label{eq124}%
\end{equation}
and%
\begin{equation}
\left.
\begin{array}
[c]{l}%
|Y_{2}(k)|=\frac{1}{T}\left \vert
{\textstyle \sum \limits_{i=k-T+1}^{k-1}}
\text{ }%
{\textstyle \sum \limits_{j=i}^{k-1}}
LM(i)S^{\mathrm{T}}(i)H\bar{\theta}(j)\right \vert \\
\leq \frac{1}{T}%
{\textstyle \sum \limits_{i=k-T+1}^{k-1}}
\text{ }%
{\textstyle \sum \limits_{j=i}^{k-1}}
\left \vert LM(i)\right \vert \left \vert S^{\mathrm{T}}(i)\right \vert \left \Vert
H\right \Vert \left \vert \bar{\theta}(j)\right \vert \\
<\frac{1}{T}%
{\textstyle \sum \limits_{i=k-T+1}^{k-1}}
\text{ }%
{\textstyle \sum \limits_{j=i}^{k-1}}
\sqrt{%
{\textstyle \sum \nolimits_{i=1}^{n}}
\frac{4l_{i}^{2}}{a_{i}^{2}}}\cdot \sqrt{%
{\textstyle \sum \nolimits_{i=1}^{n}}
a_{i}^{2}}\cdot H_{M}\cdot \varepsilon \Delta \\
=\frac{\varepsilon \Delta \cdot H_{M}}{T}\sqrt{%
{\textstyle \sum \nolimits_{i=1}^{n}}
\frac{4l_{i}^{2}}{a_{i}^{2}}}\sqrt{%
{\textstyle \sum \nolimits_{i=1}^{n}}
a_{i}^{2}}\frac{T(T-1)}{2}\\
<\frac{\varepsilon \Delta \cdot(T-1)H_{M}}{2}\sqrt{%
{\textstyle \sum \nolimits_{i=1}^{n}}
\frac{4l_{i}^{2}}{a_{i}^{2}}}\sqrt{%
{\textstyle \sum \nolimits_{i=1}^{n}}
a_{i}^{2}}\\
=\varepsilon \Delta \cdot \Delta_{3},
\end{array}
\right.  \label{eq125}%
\end{equation}
where $\Delta_{i}$ ($i\in \mathbf{I[}1,3\mathbf{]}$) are given by (\ref{eq41}).
Via (\ref{eq122}) and (\ref{eq123})-(\ref{eq125}), we have%
\begin{equation}
\left.
\begin{array}
[c]{l}%
\left \vert z(k)\right \vert \leq \left \Vert (I_{n}+\varepsilon LH)^{k-T+1}%
\right \Vert \left \vert z(T-1)\right \vert \\
+\varepsilon%
{\textstyle \sum \limits_{i=T-1}^{k-1}}
\left \Vert (I_{n}+\varepsilon LH)^{k-i-1}\right \Vert \\
\times \lbrack \left \vert LHG(i)\right \vert +\left \vert Y_{1}(i)\right \vert
+\left \vert Y_{2}(i)\right \vert ]\\
<\left \Vert (I_{n}+\varepsilon LH)^{k-T+1}\right \Vert \left \vert
z(T-1)\right \vert \\
+\varepsilon^{2}\Delta%
{\textstyle \sum \limits_{j=1}^{3}}
\Delta_{j}%
{\textstyle \sum \limits_{i=T-1}^{k-1}}
\left \Vert (I_{n}+\varepsilon LH)^{k-i-1}\right \Vert ,\text{ }k\geq T-1.
\end{array}
\right.  \label{eq131}%
\end{equation}

For deriving a bound on $(I_{n}+\varepsilon LH)^{k},$ $k\geq0,$ consider the
nominal system%
\begin{equation}
\left.
\begin{array}
[c]{l}%
z(k+1)=\left(  I_{n}+\varepsilon LH\right)  z(k)\\
=[I_{n}+\varepsilon L\left(  \bar{H}+\Delta H\right)  ]z(k),\text{ }k\geq0,
\end{array}
\right.  \label{eq127}%
\end{equation}
where we noted \textbf{A3}. Choose the Lyapunov function $V(k)=z^{\mathrm{T}%
}(k)Pz(k)$ with $P$ satisfying $I_{n}\leq P\leq pI_{n}$. Then for
$\forall \varepsilon \in(0,\varepsilon^{\ast}],$%
\begin{equation}
\left.
\begin{array}
[c]{l}%
V(k+1)-(1-\lambda \varepsilon)^{2}V(k)\\
=z^{\mathrm{T}}(k+1)Pz(k+1)-(1-\lambda \varepsilon)^{2}z^{\mathrm{T}}(k)Pz(k)\\
=z^{\mathrm{T}}(k)[I_{n}+\varepsilon L\left(  \bar{H}+\Delta H\right)
]^{\mathrm{T}}P[I_{n}+\varepsilon L\left(  \bar{H}+\Delta H\right)  ]z(k)\\
-(1-\lambda \varepsilon)^{2}z^{\mathrm{T}}(k)Pz(k)\\
=\varepsilon z^{\mathrm{T}}(k)[\left(  \bar{H}+\Delta H\right)  ^{\mathrm{T}%
}L^{\mathrm{T}}P+PL\left(  \bar{H}+\Delta H\right)  \\
+\varepsilon \left(  \bar{H}+\Delta H\right)  ^{\mathrm{T}}L^{\mathrm{T}%
}PL\left(  \bar{H}+\Delta H\right)  +\lambda(2-\lambda \varepsilon)P]z(k)\\
\leq \varepsilon z^{\mathrm{T}}(k)[\left(  \bar{H}+\Delta H\right)
^{\mathrm{T}}L^{\mathrm{T}}P+PL\left(  \bar{H}+\Delta H\right)  \\
+\varepsilon^{\ast}\left(  \bar{H}+\Delta H\right)  ^{\mathrm{T}}%
L^{\mathrm{T}}PL\left(  \bar{H}+\Delta H\right)  +2\lambda P]z(k).
\end{array}
\right.  \label{eq86}%
\end{equation}
To compensate $\Delta Hz(k)$ in (\ref{eq86}) we apply $S$-procedure, we add to
$V(k+1)-(1-\lambda \varepsilon)^{2}V(k)$ the left hand part of
\begin{equation}
\left.  \zeta \varepsilon \left(  \kappa^{2}\left \vert z(k)\right \vert
^{2}-\left \vert \Delta Hz(k)\right \vert ^{2}\right)  \geq0\right.
\label{eq87}%
\end{equation}
with some $\zeta>0.$ Then, from (\ref{eq86}) and (\ref{eq87}), we have%
\[
\left.
\begin{array}
[c]{l}%
V(k+1)-(1-\lambda \varepsilon)^{2}V(k)\\
\leq \varepsilon z^{\mathrm{T}}(k)[\left(  \bar{H}+\Delta H\right)
^{\mathrm{T}}L^{\mathrm{T}}P+PL\left(  \bar{H}+\Delta H\right)  \\
+\varepsilon^{\ast}\left(  \bar{H}+\Delta H\right)  ^{\mathrm{T}}%
L^{\mathrm{T}}PL\left(  \bar{H}+\Delta H\right)  +2\lambda P]z(k)\\
+\zeta \varepsilon \left(  \kappa^{2}\left \vert z(k)\right \vert ^{2}-\left \vert
\Delta Hz(k)\right \vert ^{2}\right)  \\
=\varepsilon \xi^{\mathrm{T}}(k)\Phi_{1}\xi(k),
\end{array}
\right.
\]
where $\xi^{\mathrm{T}}(k)=[z^{\mathrm{T}}(k),z^{\mathrm{T}}(k)(\Delta
H)^{\mathrm{T}}]$ and $\Phi_{1}$ is given by (\ref{eq126a}). Thus, if
$\Phi_{1}<0$ in (\ref{eq126a}), we have%
\[
\left.  V(k+1)\leq(1-\lambda \varepsilon)^{2}V(k),\text{ }k\geq0,\right.
\label{eq88}%
\]
which with $I_{n}\leq P\leq pI_{n}$ yields%
\[
\left.  |z(k)|^{2}\leq V(k)\leq(1-\lambda \varepsilon)^{2k}V(0)\leq
p(1-\lambda \varepsilon)^{2k}|z(0)|^{2},\right.
\]
then%
\begin{equation}
\left.  |z(k)|\leq \sqrt{p}(1-\lambda \varepsilon)^{k}|z(0)|,\text{ }%
k\geq0.\right.  \label{eq128}%
\end{equation}
On the other hand, by using the variation of constants formula for
(\ref{eq127}), we have%
\begin{equation}
\left.  z(k)=(I_{n}+\varepsilon LH)^{k}z(0),\text{ }k\geq0.\right.
\label{eq129}%
\end{equation}
By norm's definition and (\ref{eq128})-(\ref{eq129}), we find%
\begin{equation}
\left.
\begin{array}
[c]{l}%
\left \Vert (I_{n}+\varepsilon LH)^{k}\right \Vert =\sup_{|z(0)|=1}\left \vert
(I_{n}+\varepsilon LH)^{k}z(0)\right \vert \\
\overset{(\ref{eq129})}{=}\sup_{|z(0)|=1}\left \vert z(k)\right \vert
\overset{(\ref{eq128})}{\leq}\sqrt{p}(1-\lambda \varepsilon)^{k}.
\end{array}
\right.  \label{eq130}%
\end{equation}
With (\ref{eq130}), inequality (\ref{eq131}) can be continued as%
\begin{equation}
\left.
\begin{array}
[c]{l}%
|z(k)|<\sqrt{p}(1-\lambda \varepsilon)^{(k-T+1)}|z(T-1)|\\
+\varepsilon^{2}\Delta%
{\textstyle \sum \limits_{j=1}^{3}}
\Delta_{j}\sqrt{p}%
{\textstyle \sum \limits_{i=T-1}^{k-1}}
(1-\lambda \varepsilon)^{(k-i-1)}\\
=\sqrt{p}(1-\lambda \varepsilon)^{(k-T+1)}|z(T-1)|\\
+\frac{\varepsilon \Delta(\Delta_{1}+\Delta_{2}+\Delta_{3})\sqrt{p}}{\lambda
}\left[  1-(1-\lambda \varepsilon)^{(k-T+1)}\right]  \\
\leq \sqrt{p}(1-\lambda \varepsilon)^{(k-T+1)}|z(T-1)|\\
+\frac{\varepsilon \Delta(\Delta_{1}+\Delta_{2}+\Delta_{3})\sqrt{p}}{\lambda},
\end{array}
\right.  \label{eq132}%
\end{equation}
where we noted $\lambda \varepsilon \in(0,1),$ $\forall \varepsilon
\in(0,\varepsilon^{\ast}].$ Note from (\ref{eq118b}), (\ref{eq136}) and
(\ref{eq36}) that%
\[
\left.
\begin{array}
[c]{l}%
|z(T-1)|=\left \vert \tilde{\theta}(T-1)-G(T-1)\right \vert \\
\leq \left \vert \tilde{\theta}(T-1)\right \vert +\left \vert G(T-1)\right \vert \\
<|\tilde{\theta}(0)|+\left(  T-1\right)  \varepsilon \Delta+\frac
{(T-1)\varepsilon \Delta}{2}\\
=|\tilde{\theta}(0)|+\frac{3(T-1)\varepsilon \Delta}{2},
\end{array}
\right.
\]
by which, inequality (\ref{eq132}) can be continued as%
\begin{equation}
\left.  \sqrt{p}\left(  \sigma_{0}+\frac{\varepsilon^{\ast}\Delta
\lbrack2(\Delta_{1}+\Delta_{2}+\Delta_{3})+3\delta]}{2\delta}\right)
+\frac{\varepsilon^{\ast}\Delta}{2}<\sigma,\right.  \label{eq24c}%
\end{equation}
Then for $k\geq T-1,$ we have%
\begin{equation}
\left.
\begin{array}
[c]{l}%
|\tilde{\theta}(k)|=|z(k)+G(k)|\leq|z(k)|+|G(k)|\\
<\sqrt{p}(1-\lambda \varepsilon)^{k-T+1}\left(  |\tilde{\theta}(0)|+\frac
{3(T-1)\varepsilon \Delta}{2}\right)  \\
+\frac{\varepsilon \Delta(\Delta_{1}+\Delta_{2}+\Delta_{3})\sqrt{p}}{\lambda
}+\frac{(T-1)\varepsilon \Delta}{2},
\end{array}
\right.  \label{eq38}%
\end{equation}
which implies the second inequality in (\ref{eq133}) due to%
\[
\left.
\begin{array}
[c]{l}%
\sqrt{p}\left(  \sigma_{0}+\frac{\varepsilon^{\ast}\Delta \left[
3(T-1)\lambda+2(\Delta_{1}+\Delta_{2}+\Delta_{3})\right]  }{2\lambda}\right)
\\
+\frac{(T-1)\varepsilon^{\ast}\Delta}{2}<\sigma,
\end{array}
\right.
\]
namely,%
\[
\left.
\begin{array}
[c]{l}%
\sqrt{p}\left(  \sigma_{0}+\frac{\varepsilon^{\ast}\Delta \left[
3(T-1)\lambda+2(\Delta_{1}+\Delta_{2}+\Delta_{3})\right]  }{2\lambda}\right)
\\
<\sigma-\frac{(T-1)\varepsilon^{\ast}\Delta}{2},
\end{array}
\right.
\]
which, by squaring on both sides, equivalents to $\Phi_{2}<0$ in (\ref{eq126}).

Finally, we show that the conditions (\ref{eq126a})-(\ref{eq126}) guarantee the overall bound
(\ref{eq137}).

\textbf{(i)} When $k\in \mathbf{I}[0,T-1],$ since%
\[
\left.  \left \vert \tilde{\theta}\left(  0\right)  \right \vert \leq \sigma
_{0}<\sigma,\right.
\]
we assume by contradiction that for some $k \in \mathbf{I}[1,T-1]$ the
formula (\ref{eq137}) does not hold. Namely, there exists the smallest
$k^{\ast}\in \mathbf{I}[1,T-1]$ such that%
\[
\left.  \left \vert \tilde{\theta}\left(  k^{\ast}\right)  \right \vert
\geq \sigma,\text{ }\left \vert \tilde{\theta}\left(  k\right)  \right \vert
<\sigma,\text{ }k\in \mathbf{I}[0,k^{\ast}-1].\right.
\]
Thus
\[
\left.  \left \vert \bar{\theta}\left(  k\right)  \right \vert <
\varepsilon \Delta,\text{ }k\in \mathbf{I}[0,k^{\ast}-1],\right.
\]
and then%
\[
\left.
\begin{array}
[c]{l}%
\left \vert \tilde{\theta}(k+1)\right \vert =\left \vert \tilde{\theta}(0)+%
{\textstyle \sum \nolimits_{i=0}^{k}}
\bar{\theta}\left(  i\right)  \right \vert \\
\leq \left \vert \tilde{\theta}\left(  0\right)  \right \vert +%
{\textstyle \sum \nolimits_{i=0}^{k}}
\left \vert \bar{\theta}\left(  i\right)  \right \vert \\
< \sigma_{0}+\left(  T-1\right)  \varepsilon \Delta,\text{ }k\in
\mathbf{I}[0,k^{\ast}-1].
\end{array}
\right.
\]
Furthermore, the feasibility of $\Phi_{2}<\sigma$ in (\ref{eq126}) ensures
that
\[
\left \vert \tilde{\theta}\left(  k^{\ast}\right)  \right \vert < \sigma
_{0}+\left(  T-1\right)  \varepsilon^{\ast}\Delta<\sigma.
\]
This contradicts to the definition of $k^{\ast}$ such that $ \vert
\tilde{\theta}\left(  k^{\ast}\right) \vert \geq \sigma.$ Hence,
(\ref{eq137}) holds for $k\in \mathbf{I}[0,T-1].$

\textbf{(ii)} Based on the above analysis, there holds
\[
\left \vert \tilde{\theta}\left(  T-1\right)  \right \vert <\sigma.
\]
We assume by contradiction that there exist some $k>T-1$ such that $ \vert
\tilde{\theta}\left(  k\right)  \vert \geq \sigma.$ Namely, there exists
the smallest $k^{\ast}\geq T$ such that%
\[
\left \vert \tilde{\theta}\left(  k^{\ast}\right)  \right \vert \geq
\sigma,\text{ }\left \vert \tilde{\theta}\left(  k\right)  \right \vert
<\sigma,\text{ }k\in \mathbf{I}[T-1,k^{\ast}-1].
\]
Thus we have
\[
\left.  \left \vert \bar{\theta}\left(  k\right)  \right \vert <
\varepsilon \Delta,\text{ }k\in \mathbf{I}[T-1,k^{\ast}-1],\right.
\]
and then we arrive at (\ref{eq38}) with $k\in \mathbf{I}[T-1,k^{\ast}].$
Furthermore, the feasibility of $\Phi_{2}<\sigma$ in (\ref{eq126}) ensures
that $ \vert \tilde{\theta}\left(  k^{\ast
}\right) \vert  <\sigma$. This contradicts to the
definition of $k^{\ast}$ such that $ \vert \tilde{\theta}\left(  k^{\ast
}\right) \vert \geq \sigma.$ Hence (\ref{eq137}) holds for $k\geq T-1.$
The proof is finished.

\end{document}